%% file: prdc.tex
\documentclass[11pt]{article}
\usepackage{amsmath,amssymb,amsfonts,amsthm,epsfig}
\usepackage[usenames,dvipsnames]{xcolor}

\usepackage{bm,xspace}

\usepackage{cancel}
\usepackage{fullpage}
\usepackage{liyang}
\usepackage{framed}
\usepackage{verbatim}
\usepackage{enumitem}
\usepackage{array}
\usepackage{multirow}
\usepackage{afterpage}
\usepackage{mathrsfs}
\usepackage{comment}  
\usepackage{caption}
\usepackage[normalem]{ulem}

\newcommand{\rnote}[1]{\footnote{{\bf \color{red}Rocco}: {#1}}}
\newcommand{\snote}[1]{\footnote{{\bf \color{blue}Sandip}: {#1}}}
\newcommand{\xnote}[1]{\footnote{{\bf \color{red}Xi}: {#1}}}

\usepackage{caption} 

\usepackage{todonotes}

\makeatletter
\newtheorem*{rep@theorem}{\rep@title}
\newcommand{\newreptheorem}[2]{
\newenvironment{rep#1}[1]{
 \def\rep@title{#2 \ref{##1}}
 \begin{rep@theorem}\itshape}
 {\end{rep@theorem}}}
\makeatother
\theoremstyle{plain}

\newtheorem{property}[theorem]{Property}



\newcommand{\ignore}[1]{}

\def\colorful{1}

\ifnum\colorful=1

\newcommand{\blue}[1]{{{\color{blue}#1}}}
\newcommand{\red}[1]{{\color{red} {#1}}}

\newcommand{\gray}[1]{{\color{gray}{#1}}}

\fi
\ifnum\colorful=0

\newcommand{\blue}[1]{{{#1}}}
\newcommand{\red}[1]{{{#1}}}

\newcommand{\gray}[1]{{{#1}}}

\fi

\usepackage{boxedminipage}

\newreptheorem{theorem}{Theorem}
\newtheorem*{theorem*}{Theorem}
\newreptheorem{lemma}{Lemma}
\newreptheorem{proposition}{Proposition}
\newtheorem*{noclaim*}{Claim}

\newcommand{\Bin}{\mathrm{Bin}}

\newcommand{\proj}{\mathrm{restrict}}

\newcommand{\Mix}{\mathrm{Mix}}

\newcommand{\Del}{\mathrm{Del}}

\newcommand{\B}{\mathrm{B}}
\def\SD{\mathsf{D}}
\def\SumD{\mathsf{SD}}

\usepackage{dsfont}

\renewcommand{\N}{\mathds{N}}

\def\deck{\mathsf{D}}
\def\sdeck{\mathsf{SD}}

\begin{document}

\title{Beyond trace reconstruction:\\
	Population recovery from the deletion channel}

\author{
Frank Ban\thanks{Supported by NSF award CCF-1819935}\\
UC Berkeley\\
fban@berkeley.edu
\and
Xi Chen\thanks{Supported by NSF awards CCF-1703925 and IIS-1838154.} \\
Columbia University\\
xichen@cs.columbia.edu
\and
Adam Freilich\thanks{Supported by NSF award CCF-1563155.} \\
Columbia University\\
freilich.am@gmail.com
\and
Rocco A. Servedio\thanks{Supported by NSF awards CCF-1563155, CCF-1814873 and IIS-1838154.} \\
Columbia University\\
rocco@cs.columbia.edu
\and
Sandip Sinha\thanks{Supported by NSF awards CCF-1563155, CCF-1420349, CCF-1617955, CCF-1740833, CCF-1421161,\newline CCF-1714818 and Simons Foundation (\#491119).}\\
Columbia University\\
sandip@cs.columbia.edu
}

\maketitle

\thispagestyle{empty}

\input{abstract}

\newpage

\setcounter{page}{1}

\input{intro}

\input{preliminaries}



\input{worst-case-upper-bound-general-p}
\input{worst-case-lower-bound}

\begin{flushleft}
\bibliography{prdc}{}
\bibliographystyle{alpha}
\end{flushleft}

\appendix

\input{polynomial-h}

\end{document}

%% file: abstract.tex

\begin{abstract}
	\emph{Population recovery} is the problem of learning an unknown distribution over an unknown set of $n$-bit strings, given access to independent draws from the distribution that have been independently corrupted according to some noise channel. Recent work has intensively studied such problems both for the {bit-flip} noise channel and for the erasure noise channel.
	
	In this paper we initiate the study of population recovery under the \emph{deletion channel}, in which each bit $b$ is independently \emph{deleted} with some fixed probability and the surviving bits are concatenated and transmitted.  This is a far more challenging noise model than bit-flip~noise or erasure noise; indeed, even the simplest case in which the population is of size 1 (corresponding to a trivial probability distribution supported on a single string) corresponds to the \emph{trace reconstruction} problem, which is a challenging problem that has received much recent attention (see e.g.~\cite{DOS17,NP17,PZ17,HPP18,HHP18}).
		
	In this work we give algorithms and lower bounds for population recovery under the deletion channel when the population size is some value $\ell > 1$.  As our main sample complexity upper bound, we show that for any  population size $\ell = o(\log n / \log \log n)$, a population of $\ell$ strings from $\zo^n$ can be learned under deletion channel noise using $\smash{2^{n^{1/2 + o(1)}}}$ samples. On the lower bound side, we show that at least $n^{\Omega{(\ell)}}$ samples are required to perform population recovery under the deletion channel when the population size is $\ell$, for all $\ell \leq n^{1/2 - \epsilon}$.

Our upper bounds are obtained via a robust multivariate generalization of a polynomial-based analysis, due to Krasikov and Roddity \cite{KR97}, of how the $k$-deck of a bit-string uniquely identifies the string; this is a very different approach from recent algorithms for trace reconstruction (the $\ell=1$ case).
Our lower bounds build on moment-matching results of Roos~\cite{Roo00} and Daskalakis and Papadimitriou~\cite{DP15}.
\end{abstract}

\newpage

%% file: intro.tex

\section{Introduction} \label{sec:intro}

In recent years the unsupervised learning problem of \emph{population recovery} has emerged as a significant focus of research attention in theoretical computer science \cite{DRWY12,MoitraSaks13,BIMP13,LZ15,DST16,WY16,PSW17,DOS17poprec}.  In the population recovery problem there is an unknown distribution $\bX$ over an unknown set of $n$-bit strings from $\{0,1\}^n$, and the learner's job is to reconstruct a high-accuracy approximation of $\bX$ given access to noisy independent draws from $\bX$ (so each data point which the learning algorithm receives is independently generated as follows:  an $n$-bit string is drawn from $\bX$ and corrupted by some noise process, and the result is provided to the learning algorithm).  The two noise models which have chiefly been studied to date are the \emph{bit-flip} noise model, in which each coordinate is independently flipped with some fixed probability, and the \emph{erasure} noise model, in which each coordinate is independently replaced by `?' with some fixed probability.  

Since the population recovery problem was first introduced in \cite{DRWY12,WY16}, a number~of positive results and lower bounds have been obtained for different variants of the problem.  In~one popular version of the problem \cite{PSW17,DOS17poprec,MoitraSaks13}, for a particular noise model (bit-flip~or erasure) the distribution $\bX$ may be an arbitrary distribution over $\zo^n$, and the goal is to learn the distribution $\bX$ with respect to $\ell_\infty$ distance (i.e.~to output a list of strings $x^1,\dots,x^r\in \{0,1\}^n$ and associated weights $\tilde{\bX}(x^i)$ 
such that $|\bX(x^i) - \tilde{\bX}(x^i)| \leq \eps$ for all $i \in [r]$ and $\bX(x) \leq \eps$ for all $x \in \zo^n \setminus \{x^1,\dots,x^r\}$).  In another well-studied version of the problem \cite{WY16,LZ15,DST16}, which is closely related to the problems we shall consider, the distribution $\bX$ is promised to be supported on at most $\ell$ strings in $\zo^n$ (i.e.~the ``population size'' is promised to be at most $\ell$), and the goal is to output a hypothesis distribution $\tilde{\bX}$ over $\zo^n$ which has total variation distance at most $\eps$ from $\bX$. Significant progress has been made on determining the sample complexity of population recovery for both of these variants under the bit-flip and erasure noise models; we refer the interested reader to \cite{DST16,PSW17,DOS17poprec} for the current state of the art.

\medskip \noindent
{\bf This work: Population recovery from the deletion channel and its relation to trace reconstruction.}  In both the bit-flip noise model and the erasure noise model, all of the challenge in the population recovery problem stems from the fact that given a noisy draw from $\bX$ it is \emph{a priori} not clear which element of $\bX$'s support was corrupted by noise to produce the noisy draw.  Putting it another way, if the population size is promised to be $\ell=1$, then under either of these two noise models it is trivially easy to learn 
  a single unknown string 
  from noisy examples.

In this work we study population recovery under the \emph{deletion} noise model, which is  far more challenging to handle than either bit-flip noise or erasure noise.  The deletion channel is defined as follows:  when a string $x$ is passed through the deletion channel with deletion parameter $\delta$, each coordinate $x_i$ is independently deleted with probability $\delta$, the surviving coordinates are concatenated, and the resulting string (of length $\boldn' \leq n$, where $\boldn'$ is distributed as $\Bin(n,1-\delta)$) is the output of the noise process.  Intuitively, the deletion channel is challenging because given a received word obtained by passing $x$ through the $\delta$-deletion channel (often referred to as a \emph{trace} of $x$, and denoted by $\bz \leftarrow \Del_\delta(x)$), it is not clear which coordinate of $x$ gave rise to which coordinate of~$\bz$.  Indeed, in contrast with the bit-flip and erasure noise models, even if the population size is guaranteed to be $\ell=1$, the problem of recovering a single unknown string from  independent traces is a well-known and challenging open problem, known as the \emph{trace reconstruction problem} \cite{Lev01a,Lev01b,BKKM04,KM05,HMPW08,VS08,MPV14,DOS17,NP17,PZ17,HPP18,HHP18}.

 There are several motivations for the study of population recovery under the deletion noise model. One motivation is the considerable recent research interest both in the trace reconstruction problem (the $\ell=1$ case of population recovery under the deletion channel) and in population recovery problems under bit-flip and erasure models.  Further motivation comes from potential relevance of the deletion channel population recovery problem both to recovery problems in computational biology and to other topics such as DNA data storage.  Regarding biological recovery problems, considering population recovery (the $\ell>1$ case) rather than trace reconstruction (the $\ell=1$ case) relaxes the potentially unrealistic assumption that all of the received samples (of a protein sequence, DNA sequence, etc.) are derived from a single unknown target sequence rather than from multiple unknown sequences. Heuristic algorithms for population recovery-type problems have also been applied to DNA storage \cite{organick2018random}. In these settings, each string in the population comes from a DNA sequence and the noisy channel can inflict a variety of errors including bit-flips and deletions. 

Thus,  the authors feel that the time is ripe {for a theoretical study of} population recovery under the challenging deletion model.  In this paper we initiate such a study, obtaining sample complexity upper and lower bounds when the population is of size $\ell > 1.$  Before describing our results for populations of size $\ell$ (equivalently, target distributions supported on at most $\ell$ strings), we first recall known upper and lower bounds for the trace reconstruction problem ($\ell=1$) below.

\medskip
\noindent {\bf Known bounds on trace reconstruction.} The trace reconstruction problem was raised more than fifteen years ago \cite{Lev01a,Lev01b,BKKM04}, though in fact some variants of the problem~go back at least to the 1970s \cite{Kalashnik73}.  The first algorithm that provably succeeds with high~\mbox{probability} in reconstructing an arbitrary $\smash{x \in \zo^n}$ using subexponentially many traces is due to Mitzenmacher et al. \hspace{-0.08cm}\cite{HMPW08}, who showed that $\smash{2^{\tilde{O}(\sqrt{n})}}$ many traces suffice for any constant deletion~rate 
   $\delta$ bounded away from $1$. This~result was  improved in recent simultaneous and independent works of De et al.~\cite{DOS17} and Nazarov and Peres~\cite{NP17}; these papers each showed that for any constant $ \delta $ bounded away from $1$, at most $\smash{2^{O(n^{1/3})}}$ traces suffice to reconstruct any $x \in \zo^n$.\footnote{Hartung, Holden and Peres \cite{HHP18} have recently extended this result to certain more general regimes where there can be different deletion probabilities for different coordinates and symbols.}

Due to the seeming difficulty of the worst-case trace reconstruction problem (reconstructing an arbitrary $x \in \zo^n$), an average-case version of the problem (reconstructing a randomly chosen string $x \in \zo^n$), which turns out to be significantly easier in terms of sample complexity, has also received considerable attention.
A number of early works \cite{BKKM04,KM05,VS08} gave efficient algorithms that succeed for trace reconstruction of almost all $x \in \zo^n$ when the deletion rate $\delta$ is sufficiently low ($o_n(1)$ as a function of $n$). In \cite{HMPW08} Mitzenmacher et al.~gave an algorithm which uses $\poly(n)$ traces to perform average-case trace reconstruction when the deletion rate $\delta$ is at most some sufficiently small constant.  Recently the best results on average-case trace reconstruction have been significantly strengthened in works of Peres and Zhai \cite{PZ17} and Holden, Pemantle and Peres \cite{HPP18} which build on the worst-case trace reconstruction results of \cite{DOS17,NP17}.  The latter of these papers \cite{HPP18} gives an algorithm which uses $\smash{\exp((\log n)^{1/3})}$ traces to reconstruct a random $x\in \{0,1\}^n$ when  the deletion rate is any constant bounded away from 1.

In terms of lower bounds, it is easy to see that if the deletion rate $\delta$ is at least some \mbox{positive} constant, then until $\Omega(\log n)$ draws have been received there will be some bits of the target string $x$ about which no information has been received.  Improving on this simple $\Omega(\log n)$ lower bound, McGregor et al.~\cite{MPV14} established a sample complexity lower bound of $\Omega(n)$ traces for any constant deletion rate. This was recently improved by Holden and Lyons \cite{HL18} 
  to 
  $\tilde{\Omega}(n^{5/4})$. 

Summarizing, for any constant deletion probability $0 < \delta < 1$ there is currently an exponential gap between the best lower bound of $\smash{\tilde{\Omega}(n^{5/4})}$ samples and the best upper bound of $\smash{2^{O(n^{1/3})}}$ samples for trace reconstruction of an arbitrary string $x \in \zo^n.$

\subsection{Our results}

\noindent {\bf Positive result.}\ignore{ \gray{We give two main positive results which provide sample complexity upper bounds on incomparable variants of the deletion-channel population recovery problem.}}
As our main positive result, we obtain an algorithm which learns any unknown distribution $\bX$ 
supported on at most $\ell$ strings under the deletion channel.  For any constant $\ell$~(and in fact even for $\ell$ as large as 
$o(\log n/\log\log n)$, 
  its sample complexity is exponential in $\smash{{n^{1/2 + o(1)}}}$. In more detail, our main positive result is the following:

\begin{theorem} [Learning an arbitrary mixture of $\ell$ strings under the deletion channel]
\label{thm:positive}
There is an algorithm with the following performance guarantee:  Let $\bX$ be an arbitrary distribution over at most $\ell$ strings in $\zo^n$.  For any deletion rate $0 < \delta < 1$ and any accuracy parameter $\eps$, 
if the algorithm is given access to independent draws from $\bX$ that are independently corrupted with deletion noise at rate $\delta$, then the algorithm uses 
 \[
\frac{1}{\eps^2}\cdot\left({\frac 2 {1-\delta}}\right)^{ \sqrt{n}\hspace{0.04cm} \cdot 
\hspace{0.04cm}(\log n)^{O(\ell)} }
\] 
 many samples and with probability at least $0.99$ outputs a hypothesis $\tilde{\bX}$ which is supported over at most $\ell$ strings and has total variation distance at most $\eps$ from the unknown target distribution $\bX$.
\end{theorem}

\ignore{

}

It is easy to see that if the target distribution is promised to be uniform over (a multi-set~of) at most $\ell$ strings, then the algorithm of Theorem~\ref{thm:positive} can be used to exactly reconstruct the unknown multi-set.
As we explain in Section~\ref{sec:techniques}, while Theorem~\ref{thm:positive} extends prior results on trace reconstruction (the $\ell=1$ case), it is proved using very different techniques from recent works \cite{HMPW08,DOS17,NP17,PZ17,HPP18,HHP18} on trace reconstruction.  

We note that for deletion rates $\delta$ that are bounded away from 1 by a constant, the $\smash{2^{O(n^{1/3})}}$ sample complexity bounds of \cite{DOS17,NP17} for trace reconstruction are better than the $\ell=1$ case of our result. However, our bounds apply even if the deletion rate $\delta$ is very close to 1;~{in~particular, \cite{DOS17,NP17} give no results for very high deletion rates $\delta=1-o(1/\sqrt{n})$, while Theorem~\ref{thm:positive} gives a $\smash{2^{\tilde{O}(\sqrt{n})}}$ bound for  $\delta=1-1/2^{\polylog(n)}$ and a $2^{o(n)}$ bound even for $\delta$ as large as $\smash{1-1/2^{\sqrt{n}/\polylog(n)}}.$ Of course, the main feature of Theorem~\ref{thm:positive} is that it applies when $\ell > 1$ (unlike \cite{DOS17,NP17}).}\medskip

\ignore{
We note that for deletion rates $\delta$ that are bounded away from 1 by a constant, the $\smash{2^{O(n^{1/3})}}$ sample complexity bounds of \cite{DOS17,NP17} for trace reconstruction are better than the $\ell=1$ case of our result. However, our bounds apply even if the deletion rate $\delta$ is very close to 1, in particular in regimes where the \cite{DOS17,NP17} bounds do not apply.  In more detail, for $1/2 \geq 1-\delta \geq \Theta(1/\sqrt{n})$
\cite{DOS17} gives a $\smash{2^{O((n/(1-\delta))^{1/3})}}$ upper bound on sample complexity; this essentially matches our $\ell=1$ bound when $\delta=1-\Theta(1/\sqrt{n})$. But while \cite{DOS17,NP17} give no results for very high deletion rates $\delta=1-o(1/\sqrt{n})$, Theorem~\ref{thm:positive} gives a $\smash{2^{\tilde{O}(\sqrt{n})}}$ bound for any $\delta=1-1/2^{\polylog(n)}$, and indeed~a $2^{o(n)}$ bound even for $\delta$ as large as $\smash{1-1/2^{\sqrt{n}/\polylog(n)}}.$ Of course, the main feature of Theorem~\ref{thm:positive} is that it applies when $\ell > 1$ (unlike the \cite{DOS17,NP17} results).
}

\medskip
\noindent {\bf Negative result.}
Complementing the sample complexity upper bound, we obtain a lower bound on the sample complexity of population recovery.  Our lower bound shows that for a wide range of values of $\ell$, at least $n^{\Omega(\ell)}$ samples are required when the population is of size at most $\ell$.  An informal version of our lower bound is as follows (see Theorem~\ref{thm:worst-case-lower} in Section~\ref{sec:worst-case-lower} for a detailed statement):

\begin{theorem} [Sample complexity lower bound, informal statement] \label{thm:simplified-worst-case-lower} Let  $0<\delta<1$ be any constant deletion probability and suppose that $A$ is an algorithm which, when run on  samples drawn from~the $\delta$-deletion channel over an arbitrary distribution $\bX$ supported over at most $\ell \leq n^{0.499}$ many strings, with probability at least $0.51$ outputs a hypothesis distribution $\tilde{\bX}$ that has total variation distance~at most $0.49$ from the unknown target distribution $\bX$. Then $A$ must use $n^{\Omega(\ell)}$ many samples.
\end{theorem}

\section{Our techniques} \label{sec:techniques}

As noted earlier, our positive result (Theorem~\ref{thm:positive}) gives a sample complexity upper bound for the original $(\ell=1)$ trace reconstruction problem as a special case.  We remark that both of the recent $\smash{2^{O(n^{1/3})}}$ sample complexity upper bounds for the trace reconstruction problem \cite{DOS17,NP17}, as well as the earlier work of \cite{HMPW08}, employed essentially the same algorithmic approach, which is referred to in \cite{DOS17} as a ``mean-based algorithm.''  At a high level, mean-based algorithms use their samples (traces) only to compute empirical estimates of the $n$ expectations\footnote{In this context, the original unknown target string $x$ is viewed as belonging to $\bn$, and a trace $z$ obtained from $\Del_\delta(x)$ is viewed as a string in $\smash{\bits^{n'}}$ for some $n' \leq n$ with $n -n'$ zeros appended to the end.
Throughout the paper, we use $[0:n-1]=\{0,\ldots,n-1\}$ to index entries of a string of length $n$.}
\begin{equation}
\label{eq:means}
\E_{\bz \leftarrow \Del_\delta(x)}[\bz_0],\ \cdots\hspace{0.05cm}, \E_{\bz \leftarrow \Del_\delta(x)}[\bz_{n-1}]
\end{equation}
corresponding to the coordinate means of the received traces; they then only use those $n$ estimates to reconstruct the unknown target string $x$.  Both of the algorithms in \cite{DOS17,NP17}, as well~as the algorithm from \cite{HMPW08} for trace reconstruction from an arbitrary string $x$, are mean-based algorithms.  (Both \cite{DOS17} and \cite{NP17} show that their sample complexity upper bounds are essentially best possible for any mean-based trace reconstruction algorithm.)

While mean-based algorithms have led to state-of-the-art results for trace reconstruction of~a single string, this approach breaks down even for the simplest non-trivial cases of population recovery under the deletion channel.  Indeed, even when $\ell=2$ and the unknown distribution $\bX$ is promised to be uniform over two strings, it is easy to see that 
the coordinate means do not provide enough information to recover $\bX$. For example, if $(x^1,x^2)$ and $(y^1,y^2)$ are two pairs of strings whose sums (as vectors in $\R^n$) $x^1 + x^2$ and $y^1 + y^2$ are equal (such as $x^1=0^n,$ $x^2 = 1^n$, $y^1 = 0^{n/2}1^{n/2},$ $y^2 = 1^{n/2}0^{n/2}$), it is easy to see that 
the coordinate means of received traces will match perfectly: 
\[
\Ex_{\bj \in \{1,2\}}
\Ex_{\bz \leftarrow \Del_\delta(x^{\bj})}
[\bz_i]
=
\Ex_{\bj \in \{1,2\}}
\Ex_{\bz \leftarrow \Del_\delta(y^{\bj})}
[\bz_i],\quad\text{for every $i\in \{0,\ldots,n-1\}$.}
\]
Thus the mean-based approach of \cite{HMPW08,DOS17,NP17} does not suffice for even the simplest version of the population recovery problem when $\ell=2$. Indeed, our sample complexity upper bounds are obtained using a completely different approach, which we explain below.
	
\subsection{Warm-up:  A different approach to trace reconstruction (the $\ell=1$ case)}

As a warm-up to our main results, we first give a high-level explanation of how our approach can be used to obtain a simple $\smash{2^{\tilde{O}(\sqrt{n})}}$-sample algorithm for the trace reconstruction problem.  While this is a higher sample complexity than the state-of-the-art mean-based approach of \cite{DOS17,NP17} (though our approach does better for very high deletion rates, as noted earlier), our approach has the crucial advantage that it can be adapted to go beyond the $\ell=1$ case, whereas the mean-based approach cannot handle $\ell>1$ as described above.

In a nutshell, the essence of our approach is to work with \emph{subsequence frequencies} in the \emph{original string $x$} (in contrast, note that the mean-based approach uses \emph{single-coordinate frequencies} in the \emph{received traces}).   To explain further we introduce some useful terminology:  the \emph{$k$-deck} of a string $x \in$  $\zo^n$, denoted $\deck_k(x)$, is the multi-set of all ${n \choose k}$ subsequences of $x$ with length exactly $k$.  Thus, the $k$-deck encapsulates all frequency information about length-$k$ subsequences of $x$.  

A question that arises naturally in the combinatorics of words is the following:  what is the smallest value of $k$ (as a function of $n$) so that for every string $x \in \zo^n$, the $k$-deck of $x$ uniquely identifies $x$?  Despite significant investigation dating back to the 1970s \cite{Kalashnik73}, this basic quantity is still poorly understood.  Improving on earlier $k \leq n/2$ bounds of Kalashnik \cite{Kalashnik73} and Manvel et al.~\cite{MMSSS91} and a simultaneous $k=O(\sqrt{n \log n})$ bound of Scott \cite{Scott97}, Krasikov and Roddity \cite{KR97}  showed that $k=O(\sqrt{n})$ suffices.  On the lower bounds side, the best lower bound known is $\smash{k = 2^{\Omega (\sqrt{\log n} )}}$, due to Dud\'{i}k and Schulman~\cite{DS03} (improving on earlier $k=\Omega(\log n)$ lower bounds of \cite{MMSSS91} and \cite{CK97}).

The relevance of upper bounds on $k$ to the trace reconstruction problem is intuitively clear, and indeed, McGregor et al.~\cite{MPV14} observed that if the deletion rate $\delta$ is at most $1 - c \sqrt{({\log n})/ n}$, then it is trivially easy to extract a random length-$O(\sqrt{n \log n})$ subsequence of $x$ from a typical trace of $x$.  Combining this with the $k=O(\sqrt{n \log n})$ upper bound of Scott \cite{Scott97} and a straightforward sampling-based procedure (which estimates the frequency of each string in $\zo^k$  to high enough accuracy to determine its exact multiplicity in the $k$-deck), they obtained an information-theoretic sample complexity upper bound on trace reconstruction: for $\delta \leq 1 - c \sqrt{({\log n})/ n}$, at most $\smash{n^{O(\sqrt{n\log n})}}$ 
traces suffice to reconstruct any $x \in \zo^n$ with high probability.

As {an initial observation}, we slightly strengthen the \cite{MPV14} result by showing that for \emph{any} value of $\delta < 1$, an algorithm which combines sampling and dynamic programming can exactly infer the $k$-deck of an unknown string $x\in \{0,1\}^n$ with high probability using $\left({ n/( {1-\delta})}\right)^{O(k)}$ traces from $\Del_\delta(x).$  (See {Theorem~\ref{thm:traces-to-deck}} for a detailed statement and proof of a more general version of this result.)   Combining this with the \cite{KR97} upper bound $k=O(\sqrt{n})$, we get that  any string $x$ can be reconstructed from $\delta$-deletion noise using $\smash{\left( n /({1-\delta})\right)^{O(\sqrt{n})}}$ samples.

The above-outlined approach to trace reconstruction (the $\ell=1$ case of population recovery) is the starting point for our {main positive result}, Theorem~\ref{thm:positive}. In the next subsection we give a~high-level description of some of the challenges that arise in {dealing with}\ignore{extending this approach to } multiple strings and how this work overcomes them.

\subsection{Ingredients in the proof of Theorem~\ref{thm:positive}}

Recall that in the setting of Theorem~\ref{thm:positive} the unknown $\bX$ is an arbitrary distribution supported on at most $\ell$ strings $x^1,\dots,x^\ell$ in $\zo^n$. Viewing $\bX$ as a mixture of individual strings, there is a natural notion of the $k$-deck of $\bX$, which we denote by $\deck_k(\bX)$ and which is the weighted multi-set corresponding to the $\bX$-mixture of the decks $\deck_k(x^1),\dots,\deck_k(x^\ell).$\footnote{
	By a weighted-multiset we mean a multiset in which each element has a weight.
Alternatively, one can interpret (after normalization) $\deck_k(x)$ as a probability distribution
  over the $2^k$ strings in $\{0,1\}^k$ and in this case, $\deck_k(\bX)$ can be viewed as
  a probability distribution that is the $\bX$-mixture of $\deck_k(x^1),\dots,\deck_k(x^{\ell})$.}
As a result, Theorem~\ref{thm:positive} will follow~if we can show the following: if two distributions $\bX,\bY$ over $\zo^n$ (each supported on at most $\ell$ strings) have $\dtv(\bX,\bY)>\eps$, then for a not-too-large value of $k$, the $k$-decks $\deck_k(\bX)$ and $\deck_k(\bY)$ (note that these are two weighted multi-sets of strings in $\zo^k$) must be ``noticeably different.''  This is established in Lemma~\ref{maintechnical}, which is the technical heart of our upper bound.

To explain our proof of Lemma~\ref{maintechnical} it is useful to revisit the $\ell=1$ setting; the analogous (and much easier to prove) statement in this context is that given any two strings $x \neq y \in \zo^n$,  the $k$-decks $\deck_k(x)$ and $\deck_k(y)$ are not identical when $ {k \geq C \sqrt{n}}$  for some large enough constant $C$.  This is the main result of \cite{KR97} (and a similar statement, with a slightly weaker quantitative bound on $k$, is also proved in \cite{Scott97}).  Since the $k$-deck in and of itself is somewhat difficult to work with (being a multi-set over $\zo^k$), both \cite{KR97} and \cite{Scott97} work instead with the \emph{summed $k$-deck}, which we denote by $\sdeck_k(x)$ and which is simply the vector in $\N^k$ obtained by summing all ${n \choose k}$ elements of the $k$-deck $\deck_k(x)$ (recall that each element of $\deck_k(x)$ is a vector in $\zo^k$). Both \cite{KR97} and \cite{Scott97} actually show that for a suitable value of $k$, the \emph{summed} $k$-deck $\sdeck_k(x)$ uniquely identifies $x$ among all strings in $\zo^n$.  (Both papers also observe that by a simple counting argument, the smallest such $k$ is at least $\tilde{\Omega}(\sqrt{n}).$)  The \cite{KR97} proof reduces the analysis of the summed $k$-deck to an extremal problem about univariate polynomials. The key ingredient of their proof is the following result about univariate polynomials, which was established in \cite{BEK99} in their work on the Prouhet-Tarry-Escott problem:  
\begin{flushleft}
\begin{quote}Given any nonzero vector $\delta \in \{-1,0,1\}^n$, there is a univariate polynomial $p$ of degree $O(\sqrt{n})$ such that 
\begin{equation} \label{eq:PTE}
\sum_{0 \leq i < n} \delta_i \cdot p(i) \neq 0. \tag{$\dagger$}
\end{equation}
\end{quote}
\end{flushleft}
Setting $\delta=x-y\ne 0$, to finish the proof of $\sdeck_k(x)\ne \sdeck_k(y)$ when $x\ne y$ and $k\ge C\sqrt{n}$,
  \cite{KR97} shows that 
  choosing $k$ to be $\deg(p)+1$, the inequality (\ref{eq:PTE}) implies that $\sdeck_k(x)\ne \sdeck_k(y)$.

Returning to our $\ell$-string 
  setting, we remark that several challenges arise which are~not present in~the one-string setting.  To highlight one of these, due to the difficulty of analyzing the entire $k$-deck of $\bX$ it is natural to try to work with the summed $k$-deck $\sdeck_k(\bX)$ (a nonnegative vector in $\R^k$), which is obtained by
  summing all elements of the weighted multi-set $\deck_k(\bX)$. Indeed~it can be shown via a rather straightforward extension of the \cite{KR97} analysis that, when $\bX$ is uniform over $x^1,\ldots,x^\ell$, the summed $k$-deck 
  with $k= O(\sqrt{n} \log \ell)$ suffices to exactly reconstruct the \emph{sum} $x^1 + \cdots + x^\ell$ (a vector in $\N^n$).
But even for uniform distributions, a difficulty which arises is that the summed $k$-deck (even with $k=n$) cannot 
  distinguish between two 
  uniform distributions over $x^1,\dots,x^\ell$ versus $y^1,\dots,y^\ell$ that have the same coordinate-wise sums, i.e. that satisfy $x^1 + \cdots + x^\ell = y^1 + \cdots + y^\ell$.\footnote{This is conceptually similar to the inability of mean-based algorithms to handle multiple strings noted earlier.}  Indeed, considering the same example as earlier, in which $\ell=2$ and~$x^1=0^n,\ x^2 = 1^n,y^1 = 0^{n/2}1^{n/2}$ and $y^2 = 1^{n/2}0^{n/2}$, the summed $k$-deck is $({n \choose k},\dots,{n \choose k})/2 \in \R^k$
  in both cases.

At a high level our Lemma \ref{maintechnical} can be viewed as a \emph{robust} generalization of the \cite{KR97} result.  A key technical ingredient in its proof is a robust generalization of the \cite{BEK99} result to \emph{multivariate} polynomials.  (The summed $k$-deck corresponds to univariate polynomials, so at a high level our analysis involving multivariate polynomials can be viewed as how we get around the obstacle noted in the previous paragraph.)  
The proof of Lemma \ref{maintechnical} consists of three steps which we outline below. 

The first conceptual step of our argument is to show that if two support-$\ell$ distributions $\bX$ and $\bY$ over $\zo^n$  satisfy $\dtv(\bX,\bY) \ge \eps$, then there exists a subset $T\subset [0:n-1]$ of size {$d = \lfloor \log(2 \ell) \rfloor$} such that $\bX$ and $\bY$ ``differ significantly'' just on the coordinates in $T$. In particular, there is some $|T|$-bit string $c$ such that $\Pr_{\bx\sim \bX}[\hspace{0.03cm}\bx_T=c\hspace{0.03cm}]$ is significantly different from $\Pr_{\by\sim \bY} [\hspace{0.03cm}\by_T=c\hspace{0.03cm}]$,
  where we use $x_T$ to denote the restriction of a string $x\in \{0,1\}^n$ on coordinates in $T$.
(This is made precise in Lemma~\ref{large-TV-dist-implies-large-delta-max}.)
Let $\smash{\Delta:{[0:n-1]\choose d}\rightarrow \R}$ be the following function over size-$d$ subsets 
  of $[0:n-1]$:
\begin{equation}\label{DEFINEDELTA}
\Delta(S)=\Pr_{\bx\sim \bX}\big[\bx_S=c\big]-\Pr_{\by\sim \bY} \big[\by_S=c\big].
\end{equation}
Then Lemma \ref{large-TV-dist-implies-large-delta-max} implies that $\|\Delta\|_\infty$ is not too small.

The second (and central) conceptual step of our argument can be viewed as a robust generalization of the \cite{BEK99} result to $d$-variate polynomials, as alluded to earlier.
The key result giving this step, Lemma~\ref{technicallemma1}, roughly speaking states the following:

\begin{flushleft}\begin{quote}
Given the $\Delta$ as defined in (\ref{DEFINEDELTA}),
  there is 
a $d$-variate polynomial $\phi$ of not-too-high degree (roughly $\sqrt{n}$) such that\footnote{The reader who has peeked ahead to the statement of Lemma~\ref{technicallemma1} may have noticed that the lemma statement also bounds the magnitudes of coefficients of the polynomial $\phi$.  This is done for technical reasons, 
  and we skip these technical details in the high-level description here.}
\begin{equation}
\label{eq:ourPTE}
\left|\sum_{0\le t_1<\cdots<t_d<n} \phi(t_1,\ldots,t_d)\cdot \Delta\big(\{t_1,\ldots,t_d\}\big)\right|
\tag{$\dagger \dagger$}
\end{equation}
can be lower bounded in terms of $\|\Delta\|_\infty$, which is not too small by Lemma \ref{large-TV-dist-implies-large-delta-max}.
\end{quote}\end{flushleft}

The third conceptual step relates (\ref{eq:ourPTE}) to 
  the distance between the $k$-decks $\deck_k(\bX)$ and $\deck_k(\bY)$, by showing that
  if (\ref{eq:ourPTE}) is not too small then $\deck_k(\bX)$ and $\deck_k(\bY)$ must be
  ``noticeably different'' when $k$ is chosen to be $\deg(\phi)+d$. 
We refer the reader to Lemma~\ref{technicallemma2}.
At a high level this is analogous to, but technically more involved than, the \cite{KR97} proof that the inequality 
  (\ref{eq:PTE}) for $\delta=x-y$~implies~that $\sdeck_k(x)\ne \sdeck_k(y)$ with $k=\deg(p)+1$.
Lemma \ref{maintechnical} then follows by combining all three steps,~i.e. $\dtv(\bX,\bY)$ being large implies that $\deck_k(\bX)$ is 
  ``noticeably different'' from $\deck_k(\bY)$ for $k$ that is roughly $\sqrt{n}$.
Below we outline the main ingredients needed in the second step.

In the search for a low-degree polynomial $\phi$ such that the sum in (\ref{eq:ourPTE}) has large magnitude,
  it is natural to define $\phi(t_1,\ldots,t_d)$ by first projecting $(t_1,\ldots,t_d)$ to
  a line and then applying a univariate polynomial similar to the $p$ used in (\ref{eq:PTE}).
To make this more precise, we will look for $\phi$ of the form
\begin{equation}\label{hehehehehe}
\phi(t_1,\ldots,t_d)=f\big(w_1t_1+\cdots+w_dt_d\big),
\end{equation}
where $w_1,\ldots,w_d$ are positive integers (so the line is along the direction $w=(w_1,\ldots,w_d)$)
  and $f$ is a low-degree univariate polynomial to be specified later.
With (\ref{hehehehehe}), we rewrite  the sum in (\ref{eq:ourPTE}) as
\begin{equation}\label{heheeqhehe}
\sum_{0\le t_1<\cdots<t_d<n} \phi(t_1,\ldots,t_d)\cdot \Delta\big(\{t_1,\ldots,t_d\}\big)
=\sum_{b=0}^{nd\hspace{0.05cm}\|w\|_\infty} f(b)\cdot \Gamma(b),
\end{equation}
where $\Gamma(b)$ is the sum of $\Delta(T)$ over all $d$-subsets $T=\{t_1,\ldots,t_d\}$
  such that $0 \leq t_1<\cdots<t_d < n$ and $b=w_1t_1+\cdots$ $+ w_dt_d$.
Comparing (\ref{heheeqhehe}) with (\ref{eq:PTE}), our goal would follow directly from the \cite{BEK99} result by choosing
  $f$ to be $p$ if $\Gamma$ 
  is nonzero and takes values in $\{-1,0,1\}$ (or even $\{-c,0,c\}$ for some not too small $c>0$).
However, the main difficulty we 
encounter is that $\Gamma$ is much more complex than the $\{-1,0,1\}^n$ vectors that can be handled by
  techniques of \cite{BEK99}; for example, $\Gamma$ in general 
  may contain a large number (depending on $n$) of distinct values. 
   
There are three ingredients we use in choosing $w_1,\ldots,w_d$ and $f$ to
  overcome this difficulty: 
\begin{flushleft}\begin{enumerate}
\item[(A)] We first observe that $\Delta$ has a combinatorial 
    ``rectangular'' structure, which implies that the support of $\Delta$ can be partitioned into 
    a small number of sets $\calS_1,\calS_2,\ldots$ (each element of $\calS_a$ is a size-$d$ subset of $[0:n-1]$)
    such that all $T\in \calS_a$ share the same value of $\Delta(T)$ and 
    there is a set $T_a\in \calS_a$ that is \emph{dominated}\footnote{Given two size-$d$ 
    subsets $S = \{s_1,\ldots,s_d\}$ and $T = \{t_1,\ldots,t_d\}$ of $[0 : n-1]$ with $s_1 <
\cdots < s_d$ and $t_1 <\cdots < t_d$, we say that $S$ is dominated by $T$ if $s_i\le t_i$ for all $i$.}
    by every $T\in \calS_a$.
We refer to $T_a$ as the \emph{anchor set} of $\calS_a$. 
This is made precise in Lemma \ref{coverlemma}. 
Moreover, we show in Lemma \ref{coverlemma2} that the collections $\calS_a$ 
	can be divided into
  an \emph{even smaller} number of groups such that, for any $\calS_a,\calS_{a'}$ that belong to the same group, the ratio of $|\Delta(T_a)|$ and $|\Delta(T_{a'})|$ is bounded from above by a small number.
\vspace{0.2cm}
\item[(B)] Next we observe that when $w_1,\ldots,w_d$ are drawn from a suitable distribution,\ignore{ and independently from $\{1,\ldots,L\}$
  for some positive integer $L$ that is not too large,} all anchor sets in (A) have distinct images after 
  the projection. (See Claim \ref{blue1}.)
We fix such a tuple $(w_1,\ldots,w_d)$. (A) and (B) together are then used to obtain (see Lemma \ref{starlem}) a strong
  structural characterization of $\Gamma$. 
\item[(C)] Finally we define a new univariate polynomial $f$ based on Chebyshev polynomials and the
construction of $p$ in \cite{BEK99}. (See Lemma \ref{uuuuu}.)
The characterization of $\Gamma$ and properties of $f$ are then combined to 
  finish the proof by showing that the sum in (\ref{eq:ourPTE}) has {not too small} magnitude when we apply the polynomial $\phi$
  given in (\ref{hehehehehe}).
\end{enumerate}\end{flushleft}







\subsection{Our lower bounds}

We begin by recalling the $\Omega(n)$ lower bound of McGregor et al.~\cite{MPV14}.  This lower bound is obtained via a simple analysis of the two distributions of traces resulting from the two strings $x^1 = 0^{n/2}10^{n/2-1}$ and $x^2=0^{n/2-1}10^{n/2}$.  The starting point of the \cite{MPV14} analysis is the observation that under the $\delta$-deletion channel, conditioned on the sole ``1'' coordinate being retained, the distribution of a trace of $x^1$ corresponds to $(\ba,\bb)$ where $\ba$ and $\bb$ are independent draws from $\Bin(n/2,1-\delta)$ and $\Bin(n/2-1,1-\delta)$ respectively, whereas the distribution of a trace of $x^2$ corresponds to $(\bb,\ba).$  \cite{MPV14} used this to show that the squared Hellinger distance between these two distributions of traces is $O(1/n)$, and in turn use this squared Hellinger distance bound to infer an $\Omega(n)$ sample complexity lower bound for determining whether a collection of received traces came from $x^1$ or from $x^2$.

Our lower bound approach may be viewed as an extension of the \cite{MPV14} lower bound to \emph{mixtures} of distributions similar to the ones they consider.  
The high-level idea of our lower bound proof is as follows:  we show that there exist two distributions $\bX,\bY$ over $\zo^n$ (in fact, over $n$-bit strings with precisely one 1) which have disjoint supports, each of size at most ${2\ell}$, but are such that the total variation distance
$\dtv(\Del_{\delta}(\bX),\Del_{\delta}(\bY))$,
between traces of strings drawn from $\bX$ versus traces of strings drawn from $\bY$, is very small.  This is easily seen to imply Theorem~\ref{thm:simplified-worst-case-lower}.

For simplicity in introducing the main ideas of our analysis, in this expository overview we will first consider an ``$n=+\infty$'' version of our population recovery scenario.  We begin by considering the distribution $\Del_{\delta}(\tilde{e}_{m+i})$ where $m$ is some fixed value and $\tilde{e}_{m+i}$ is an infinite string with a single $1$ in position $m+i$ and all other coordinates $0$.  A $\delta$ fraction of the outcomes of $\Del_{\delta}(\tilde{e}_{m+i})$ are the infinite all-$0$ string, which conveys no information.  The other $1 - \delta$ fraction 
of the outcomes each have precisely one 1, occurring in position $1+\ba$ where $\ba$ is distributed according to the binomial distribution $\Bin(m+i,1-\delta)$.  In this infinite-$n$ setting, two distributions $\bX,\bY$ over strings of the form $\tilde{e}_{m+i}$ with disjoint supports\ignore{ \blue{$S,T \subset \{e_{m+i} \, | 0 \leq i \leq 2k\}$}}
correspond to two mixtures of distinct binomial distributions (all with second parameter $1-\delta$, but with a set of first parameters in the first mixture that is disjoint from the set of first parameters in the second mixture).  The animating idea behind our construction and analysis is that it is possible for two distinct mixtures of binomials like this to be very close to each other in total variation distance.\footnote{We remark that our actual scenario is more complicated than this idealized version because $n$ is a finite value rather than $+\infty$. For $n=2m+1$, this means that a received trace $0^{\ba} 1 0^{\bb}$ which contains a 1 and came from $\Del_{\delta}(e_{m+i})$ provides a pair of values $(\ba,\bb)$ where $\ba$ is distributed according to $\Bin(m+i,{\rho})$ and $\bb$ is independently distributed according to $\Bin(m-i,{\rho})$ where $\rho=1-\delta$ is the retention probability. This second value $\bb$ provides additional information which is not present in the $n=+\infty$ version of the problem, and this makes it more challenging and more technically involved to prove a lower bound. We deal with these issues in Section~\ref{sec:tvdub}.}

In order to show that two distinct mixtures of binomial distributions as described above can be very close to each other in total variation distance, our lower bounds employ technical machinery due to Roos~\cite{Roo00} and Daskalakis and Papadimitriou~\cite{DP15}. Roos~\cite{Roo00} developed a ``Krawtchouk expansion'' which provides an \emph{exact} expression for the probability that a Poisson binomial distribution (a sum of $n$ independent Bernoulli random variables with expectations $p_1,\dots,p_n$) puts on any given outcome in $\{0,1, \cdots,n\}$.
Daskalakis and Papadimitriou \cite{DP15} used Roos's Krawtchouk expansion to show that under mild technical conditions, low-order moments of any Poisson binomial distribution essentially determine the entire distribution. In more detail, their main result is that if $\bX,\bY$ are two Poisson binomial distributions (satisfying mild technical conditions) whose $t$-th moments match, i.e.~$\E[\bX^t]=\E[\bY^t]$ for $t=1,\dots,O(\log(1/\eps))$, then the total variation distance between $\bX$ and $\bY$ is at most $\eps.$  

Our analysis proceeds in two main steps. In the first step, we show that there exist two mixtures of pairs of binomial distributions, which we denote by $\bD_{S}$ and $\bD_{T}$, with certain desirable properties. $S$ and $T$ are both subsets of $\{0,\dots, 2 \ell\}$, and $\bD_S$ is a certain mixture of pairs of binomial distributions $(\Bin(n/2+i,1-\delta), \Bin(n/2-i,1-\delta))$ for $i \in S$ while 
$\bD_T$ is a certain mixture of pairs of binomial distributions $(\Bin(n/2+j,1-\delta), \Bin(n/2-j,1-\delta))$ for $j \in T$. We establish the existence of \emph{disjoint} sets $S,T$ such that the resulting mixtures $\bD_S$ and $\bD_T$ have matching $t$-th moments for all $t=1,\dots,\ell.$
This is proved using known algebraic expressions for the moments of binomial distributions and simple linear algebraic arguments.  In the second main step, we extend the analysis of Daskalakis and Papadimitriou \cite{DP15} and apply this extension to our setting, in which we are dealing with mixtures of (pairs of) binomial distributions (as opposed to their and Roos's setting of Poisson binomial distributions).  We show that the matching first $\ell$ moments of $\bD_S$ and $\bD_T$ imply that the distributions $\Del_\delta(\bX)$ and $\Del_\delta(\bY)$ are very close, where $\bX$ corresponds to the mixture of Hamming-weight-one strings in $\zo^n$ corresponding to $\bD_S$ and $\bY$ likewise corresponds to the mixture of Hamming-weight-one strings corresponding to $\bD_T.$  (In fact, in our setting having $\ell$ matching moments leads to $n^{-\Omega(\ell)}$-closeness in total variation distance, whereas in \cite{DP15} the resulting closeness from $\ell$ matching moments was $2^{-\Omega(\ell)}.$)

We close this subsection by observing that while the results of \cite{Roo00,DP15} were used in a crucial way in subsequent work of Daskalakis et al.~\cite{DaDS15} to obtain a sample complexity \emph{upper bound} on learning Poisson binomial distributions, in our context we use these results to obtain a sample complexity \emph{lower bound} for population recovery.  Intuitively, the difference is that in the \cite{DaDS15} scenario of learning an unknown Poisson binomial distribution, there is no noise process affecting the samples: the learning algorithm is assumed to directly receive draws from the underlying Poisson binomial distribution being learned.  In such a noise-free setting, the existence of a small $\eps$-cover for the space of all Poisson binomial distributions (which is established in \cite{DP15} as a consequence of their moment-matching result) means, at least on a conceptual level, that a learning algorithm ``need only search a small space of candidates'' to find a high-accuracy hypothesis.  In contrast, in our context of deletion-channel noise, our arguments show that it is possible for two underlying true distributions $\bX,\bY$ over $\zo^n$ to be very different (indeed, to have disjoint supports) but to be such that their deletion-noise-corrupted versions have low-order moments which match each other exactly.  In this scenario, the \cite{Roo00,DP15} results can be used to show that the variation distance between the two distributions of noisy samples received by the learner is very small, and this gives a sample complexity lower bound for distinguishing $\bX$ and $\bY$ on the basis of such noisy samples.

%% file: preliminaries.tex


\def\tuples{\Gamma_{n,d}}

\section{Preliminaries} \label{sec:preliminaries}

\noindent {\bf Notation.}
Given a nonnegative integer $n$, we write $[n]$ to denote $\{1,\ldots,n\}$.
Given integers $a\le b$ we write $[a:b]$ to denote $\{a,\ldots,b\}$. 
It will be convenient for us to index a binary string~$x \in \zo^n$
  using $[0:n-1]$ as $x=(x_0,\dots,x_{n-1})$.
Given a vector $v=(v_1,\ldots,v_d) \in \R^d$, we write $\|v\|_\infty$ to denote $\max_{i \in [d]}
 |v_i|.$
Given a function $\Delta:A\rightarrow \R$ over a finite domain $A$, we 
  write $\|\Delta\|_\infty=\max_{a\in A} |\Delta(a)|$.
{Given a polynomial $p$ (which may be univariate or multivariate), we write $\|p\|_1$ to denote the sum of magnitudes of $p$'s coefficients.} All logarithms and exponents are binary (base 2) unless otherwise specified.
\medskip

\noindent {\bf Distributions.}
We use bold font letters to denote probability distributions and 
  random variables, which should be clear from the context.
We write ``$\bx \sim \bX$'' to indicate that random variable~$\bx$~is 
  distributed according to distribution $\bX$.
The total variation distance between two distributions $\bX$ and $\tilde{\bX}$ over a finite set $\calX$ is 
  defined as
\[
\dtv(\bX,\tilde{\bX} ) = {\frac 1 2} \sum_{x \in \calX} \big|\bX(x) - \tilde{\bX}(x)\big|,
\]
where $\bX(x)$ denotes the amount of probability mass that the distribution $\bX$ puts on outcome $x$.\medskip

\noindent{\bf Population recovery from the deletion channel.}
Throughout this paper the parameter $0 <$ $\delta < 1$ denotes the \emph{deletion probability}.  Given a string $x \in \zo^n$, we write $\Del_\delta(x)$ to denote the distribution of a random trace of $x$ after it has been passed through the $\delta$-deletion channel (so the distribution $\Del_\delta(x)$ is supported on $\zo^{\leq n}$).  Recall that a random trace $\by \sim \Del_\delta(x)$ is obtained by independently deleting each bit of $x$ with probability $\delta$ and concatenating the surviving bits.\hspace{0.05cm}\footnote{For simplicity in this work we assume that the deletion probability $\delta$ is known to the learning algorithm.  We~note that it is possible to obtain a high-accuracy estimate of $\delta$ simply by measuring the average length of traces received from the deletion channel.}

We now define the problem of population recovery from the deletion channel that we will study in this paper.  In this problem the goal is to learn an unknown \emph{target distribution} $\bX$ supported~on at most $\ell$ strings from 
$\zo^n$.  
 {The learning algorithm has access to independent samples, each~of which is generated independently by first drawing  a string $\bx\sim \bX$ and then outputting a trace from $\Del_\delta(\bx)$.} For conciseness we write $\Del_\delta(\bX)$ to denote this distribution.\ignore{ $\Mix(\bX; \Del_\delta(x^1),\dots,\Del_\delta(x^k)).$}
The goal for the learning algorithm is to output with high probability (say at least $0.99$) a \emph{hypothesis distribution} $\tilde{\bX}$ for $\bX$ which is $\eps$-accurate in total variation distance:
$\dtv(\bX,\tilde{\bX})\le \eps$. We are interested in the number of samples needed for this learning task in terms of $n$, $\ell$, $\eps$ and $\delta$.
\medskip

\noindent {\bf Decks.}
Given a subset $T=\{t_1,\ldots,t_k\} \subseteq [0:n-1]$ of size $k$ with 
  $t_1<\cdots<t_k$, and two strings $v \in \zo^k$, $x \in \zo^n$, we say that \emph{$v$ matches $x$ at $T$} if $x_T = v$, where $x_T=(x_{t_1},\ldots,x_{t_k})\in \{0,1\}^k$  denotes the string $x$ restricted to positions in $T$. We say that the \emph{number of occurrences of $v$ in $x$} is the number of size-$k$ subsets $T \subseteq [0:n-1]$ such that $v$ matches $x$ at $T$, and we write $\#(v,x)$ to denote this quantity.
Given a distribution $\bX$ over $\{0,1\}^n$, we write $\#(v,\bX)$ to denote the 
  expected number of occurrences of $v$ in $\bx\sim \bX$, i.e. $$\#(v,\bX)=\Ex_{\bx\sim\bX} \big[\#(v,\bx)\big].$$

Given a string $x \in \zo^n$, we write $\deck_k(x)$ to denote the (\emph{normalized}\footnote{It will be more
  convenient for us to use the notion of (normalized) $k$-decks defined here; note that we can recover
  from it the multi-set of all  subsequences of $x$ with length $k$, and vice versa.}) \emph{$k$-deck} of $x$. 
This is a $2^k$-dimensional vector indexed by strings $v\in \{0,1\}^k$ such that 
$$
\big(\deck_k(x)\big)_v = \frac{\#(v,x)}{{n \choose k}}.
$$
So $\deck_k(x)$ is a nonnegative vector that sums to $1$.
Similarly, for a distribution $\bX$ over strings from $\{0,1\}^n$, we write $\deck_k(\bX)$ to denote 
  the (\emph{normalized}\footnote{Similarly, the (normalized) $k$-deck here is 
  equivalent to the weighted multi-set version used in the introduction up to a simple rescaling.}) \emph{$k$-deck} of $\bX$, given by
\[
\big(\deck_k(\bX)\big)_v=\frac{\#(v,\bX)}{{n\choose k}}, 
\]
for each $v\in \{0,1\}^k$. So $\deck_k(\bX)$ is also a $2^k$-dimensional nonnegative vector that sums to $1$.

%% file: worst-case-upper-bound-general-p.tex


\section{Upper bounds for distributions supported on at most $\ell$ strings}\ignore{mixtures over $\ell$ strings} \label{sec:worst-case-UB-general-p}

Our goal is to prove Theorem \ref{thm:upper-bound-uniform}, which is restated below:

\begin{theorem} \label{thm:upper-bound-uniform}
	There is an algorithm $A$ which has the following performance~\mbox{guarantee:} 
	For any~{distribution}
	$\bX$ supported over at most $\ell$ strings in $\zo^n$, if $A$ is given 
 \begin{equation}\label{mainbound}
 \frac{1}{\eps^2}\cdot \left(\frac{2}{1-\delta}\right)^{\sqrt{n}\hspace{0.04cm}\cdot\hspace{0.04cm} 
   (\log n)^{O(\ell)}}
\end{equation} 
many samples from $\Del_\delta(\bX)$, then with probability at least 0.99  the algorithm outputs a probability distribution $\tilde{\bX}$ supported over at most $\ell$ strings
	  such that $\dtv(\bX,\tilde{\bX})\le \eps$. 
\end{theorem}

\ignore{We first recall some terminology and prove two lemmas about projections in Section \ref{sec:proj}.}
In Section~\ref{sec:proj} we introduce the notion of a \emph{restriction}, which is a ``local view'' of a distribution~$\bX$ confined to a specific subset of coordinates and a specific outcome for those coordinates.  We then provide some terminology and prove {three} useful lemmas about restrictions in Section~\ref{sec:proj}.
Next in Section \ref{algorithmsec} we describe the algorithm $A$, 
  state our main technical lemma, Lemma \ref{maintechnical}, and 
  use~it to prove the correctness of algorithm $A$.
We prove Lemma  \ref{maintechnical} in Sections \ref{prooftechnical1} and \ref{prooftechnical2}.
\medskip

\noindent {\bf Notational convention.}  Our argument below involves many integer-valued index variables which take values in a range of different intervals.  To help the reader keep track, we will use the following convention (the values $L$ and $m$ will be defined later):

\begin{itemize}
\item $s,t,s_1,t_1,\dots$ will denote an index ranging over $[0:n-1]$;\vspace{-0.1cm}
\item $j,j_1,\dots$ will denote an index ranging over $[0:k-1]$;\vspace{-0.1cm}
\item $a,a',a_1,\dots$ will denote an index ranging over $[L]$;\vspace{-0.1cm} 
\item $b,b',b_1,\dots$ will denote an index ranging over $[0:m]$;\vspace{-0.1cm}
\item $i,i_1,\dots$, $\alpha,\alpha_1,\ldots$ and $\beta,\beta_1,\ldots$ will denote an index in all other places.\vspace{0.2cm}
\end{itemize}
	
\subsection{Restrictions}\label{sec:proj}	
	 
Let $\bX$ be a distribution over strings from $\{0,1\}^n$ and let $d\in [n]$ be a parameter {(which should be 
  thought of as quite small; we will  set $d= O(\log \ell)$ below).} 
Given a size-$d$ subset $T = \{t_1,\dots,t_d\}$ of $[0:n-1]$ with $0\le t_1<\cdots<t_d<n$
  and a string $c\in \{0,1\}^d$,
  we define
$$
\proj(\bX,T,c):=\Pr_{\bx\sim \bX} \big[(\bx_{t_1},\ldots,\bx_{t_d})=c\hspace{0.05cm}\big],
$$
 {the probability that a draw of $\bx \sim \bX$ matches $c$ in the coordinates of $T$.}

Let $\bX$ and $\bY$ be two distributions, each supported over at most $\ell$ strings from $\zo^n$.
Our first lemma shows that if $\dtv(\bX,\bY)$ is large, then there are a size-$d$ subset $T$ and a string
  $c\in \zo^d$ with $d=\lfloor \log(2\ell)\rfloor$ such that there is a reasonably big gap between
  $\proj(\bX,T,c)$ and $\proj(\bY,T,c)$.

 \begin{lemma} \label{large-TV-dist-implies-large-delta-max}
Let $\bX$ and $\bY$ be two distributions, each supported over at most $\ell$ strings from $\zo^n$.
Then there exist a size-$d$ subset $T$ of $[0:n-1]$ and a string $c\in \{0,1\}^d$ with $d=\lfloor \log (2\ell)\rfloor$
  such that 
$$
\Big|\hspace{0.03cm}\proj(\bX,T,c)-\proj(\bY,T,c)\hspace{0.03cm}\Big|
  \ge \frac{\dtv(\bX,\bY)}{\ell^{O(\ell)}}.
$$
\end{lemma}
\begin{proof}
 Let $\supp(\bX)\cup \supp(\bY)= \{z^1,\ldots,z^{\ell'}\}$ for some $\ell'\le 2\ell$.
For each $i\in [\ell']$, let $p_i\ge 0$ be the magnitude of the difference 
  between the probabilities of $z^i$ in $\bX$ and in $\bY$.
Let $\eps=\dtv(\bX,\bY)$. Then by definition we have $\sum_i p_i=2\eps$. Without loss of generality 
  we assume that $p_1\ge \cdots\ge p_{\ell'}\ge 0$~and prove the following claim (where we set $p_{\ell'+1}=0$ by default for convenience):

\begin{claim}
There exists an $i^*\in [\ell']$ such that $p_{i^*}\ge \eps/(4\ell)^{\ell'}$ and $p_{i^*+1}\le p_{i^*}/(4\ell)$.
\end{claim}
\begin{proof}
First we notice that $p_1\ge \eps/\ell$ given that $\sum_i p_i=2\eps$ and $\ell'\le 2\ell$.
Now given that the $p_i$'s~are nonnegative, there exists an $i\in [\ell']$ (e.g., by taking $i=\ell'$)
  such that $p_{i+1}\le p_i/(4\ell)$.
Take $i^*$ to be the smallest such index $i$.
Then we have
$$
\frac{p_{i^*}}{p_1}=\frac{p_{i^*}}{p_{i^*-1}}\cdots \frac{p_2}{p_1}>\frac{1}{(4\ell)^{i^*-1}}
$$
by the choice of $i^*$ as the smallest such index. As a result, we have $$p_{i^*}\ge \frac{\eps}{(4\ell)^{i^*}}\ge \frac{\eps}{(4\ell)^{\ell'}}.$$
This finishes the proof of the claim.
\end{proof}


Let $i^*\in [\ell']$ be the integer given by the claim above, and 
  we consider the first $i^*$ strings $z^1, \ldots,$ $z^{i^*}$.
Given that $i^*\le \ell'\le 2\ell$, there exist a $d$-subset $T$ of $[0:n-1]$
  with $d=\lfloor \log(2\ell)\rfloor$, a string $c\in \{0,1\}^d$
  and an $i'\le i^*$ such that
  the restriction of $z^{i'}$ matches $c$ but the restriction of $z^i$ does not match $c$ for
  any other $i\le i^*$.
(This can be achieved by repeatedly selecting a coordinate that splits the remaining strings into two nonempty subsets and setting $c$ to reduce the size by at least half each time.)
Using properties of $i^*$ given in the claim above, we have 


$$
\Big|\hspace{0.03cm}\proj(\bX,T,c)-\proj(\bY,T,c)\hspace{0.03cm}\Big|
  \ge p_{i^*}-\sum_{i>i^*}p_i\ge p_{i^*}-2\ell\cdot \frac{p_{i^*}}{4\ell}=\frac{p_{i^*}}{2}
    \ge \frac{\eps}{\ell^{O(\ell)}}.
$$
This finishes the proof of the lemma. 
\end{proof} 

Given two size-$d$ subsets $S=\{s_1,\ldots,s_d\}$ and $T=\{t_1 ,\ldots,t_d \}$
  of $[0:n-1]$ with $s_1<\cdots<s_d$ and $t_1 <\cdots<t_d $,
  we say that $S$ is \emph{dominated} by $T$ if
  $s_i\le t_i$ for every $i\in [d]$.
Let $\smash{\Delta:{[0:n-1]\choose d}\rightarrow \R}$ be a function
  over size-$d$ subsets of $[0:n-1]$.
We use $\supp(\Delta)$ to denote the set of subsets $T$ with $\Delta(T)\ne 0$.
We need the following definitions of a \emph{cover} and a \emph{group cover} of such a function $\Delta$.
\begin{definition}[Covers and group covers]
We say that a function $\Delta:{[0:n-1]\choose d}\rightarrow \R$ has an \emph{$L$-cover} {$\{(T_a,\calS_a):a\in [L]\}$} 
for some $L\ge 0$ if
\begin{enumerate}
\item $\calS_1,\ldots,\calS_L$ form an $L$-way partition of $\supp(\Delta)$;
\item $T_a\in \calS_a$ for each $a\in [L]$;
\item $\Delta(T)=\Delta(T_a)$ for every $T\in \calS_a$; and
\item $T_a$ is dominated by every $T\in \calS_a$.
\end{enumerate}
 {We refer to the set $T_a$ as the \emph{anchor set} of the collection $\calS_a$.}

 {Furthermore we say that $\Delta$ has an \emph{$(L,q,\lambda)$-group cover} if 
  $\Delta$ has an $L$-cover $\{(T_a,\calS_a):a\in [L]\}$ and a $q$-way partition of
  $[L]$ into $A_1,\ldots,A_q$ such that for each $i \in [q]$, for all $a,a'\in A_i$ we have
$$
\frac{|\Delta(T_a)|}{|\Delta(T_{a'})|}\le \lambda.
$$}
\end{definition}

Given distributions $\bX$ and $\bY$ over strings from $\zo^n$ and a string $c\in \{0,1\}^d$,
  we write $\Delta_{\bX,\bY,c}$ to denote the function over size-$d$ subsets of $[0:n-1]$
  that maps a size-$d$ subset $T$ to
$$
\Delta_{\bX,\bY,c}(T):=\proj(\bX,T,c)-\proj(\bY,T,c).
$$

The second lemma shows that when $d$ and the supports of $\bX,\bY$ are  small, 
  the function $\Delta_{\bX,\bY,c}$ has 
  a small cover for any string $c\in \{0,1\}^d$.
Taking as an example when $\ell=d=2$ and $\supp(\bX)=\{x^1,x^2\}$,
  we have that $\proj(\bX,S,c)=\proj(\bX,T,c)$ 
  if $x^1_S=x^1_T$ and $x^2_S=x^2_T$ 
  (note that this is a sufficient but not necessary condition in general). 
Letting $S=\{s_1,s_2\}$ for some $s_1<s_2$ and $T=\{t_1,t_2\}$ for some $t_1 < t_2$,
  this condition can be written equivalently as 
$$
(x^1_{s_1},x^2_{s_1})= (x^1_{t_1},x^2_{t_1})\quad\text{and}\quad
(x^1_{s_2},x^2_{s_2})= (x^1_{t_2},x^2_{t_2}).
$$
This implies that $\proj(\bX,\cdot ,c)$, as a function over size-$2$ subsets,
  has the following combinatorial ``rectangular'' structure:
one can partition indices $t\in [0:n-1]$ into four~types~{00,01,10,11} \mbox{according} to values of 
  $x^1_t$ and $x^2_t$; this induces a partition of all size-${2}$ subsets into $16$ ``rectangles,'' 
  \footnote{Strictly speaking, these are not rectangles
  since we always need to order indices of a subset in ascending order.} 
  where $S=\{s_1<s_2\}$ and $T=\{t_1<t_2\}$ belong to the same ``rectangle''
  iff the type of $s_1$ is the same as that of $t_1$ and the type of $s_2$ is the same as that of $t_2$. It follows that all $T$ in the same ``rectangle'' share the same value $\proj(\bX,T,c)$.  
We use this observation to obtain a small cover for $\Delta_{\bX,\bY,c}$.


\begin{lemma}\label{coverlemma}
Let $\bX$ and $\bY$ be two distributions, each supported over at most $\ell$ strings from $\zo^n$.
For any $d\in [n]$ and any string $c\in \{0,1\}^d$, $\Delta_{\bX,\bY,c}$ has an $L$-cover for some $L\le 2^{2d\ell}$.
\end{lemma}
\begin{proof}
Suppose that $\bX$ is supported on $x^1,\ldots,x^{\ell'}$ and $\bY$ is supported on $y^1,\ldots,y^{\ell''}$ 
  with $\ell',\ell'' \le \ell$.
We say an index $t\in [0:n-1]$ is of 
  \emph{type}-$(u,v)$, where $u\in \{0,1\}^{\ell'}$ and $v\in \{0,1\}^{\ell''}$, if
$$
(x^1_i,\ldots,x^{\ell'}_i)=u\quad\text{and}\quad
(y^1_i,\ldots,y^{ {\ell''}}_i)=v.
$$
This allows us to classify size-$d$ subsets of $[0:n-1]$ into
  at most $\smash{ (2^{\ell'+\ell''})^d\le 2^{2d\ell}}$ many equivalence classes:
  $S\sim T$ if $S=\{s_1,\ldots,s_d\}$ with $s_1<\cdots<s_d$ and $T=\{t_1,\ldots,t_d\}$ 
  with $t_1<\cdots<t_d$ are such that $s_i$ and $t_i$ are of the same type for all $i\in [d]$.
   
Let $\calS_a$ be a nonempty equivalence class of $\sim$ such that $S=\{s_1,\ldots,s_d\}\in \calS_a$ if $s_1<\cdots<s_d$ and $s_i$ has type-$(u^{(i)},v^{(i)})$ for each $i\in [d]$.
It follows from the definition of $\sim$ that 
  all $S\in \calS_a$~have the same $\proj(\bX,S,c)$ and $\proj(\bY,S,c)$, 
  and hence the same value of 
   $\Delta_{\bX,\bY,c}(S)$.
Moreover, we let $T_a=\{t_1,\ldots,t_d\}$~be the following set:
  $t_1$ is the smallest index of type-$(u^{(1)},v^{(1)})$ and for each $i$ from $2$ to $d$,
  $t_i$ is the smallest index that is larger than $t_{i-1}$ and has type-$(u^{(i)},v^{(i)})$.
Because $\calS_a$ is nonempty, $T_a$ is well defined and it is easy to 
  verify that $T_a$ is dominated by every $S\in \calS_a$.
As a result, $\Delta_{\bX,\bY,c}$ has the following  $L$-cover:
$$
\big\{(T_a,\calS_a):\text{$\calS_a$ is nonempty and $\Delta_{\bX,\bY,c}(T_a)\ne 0$}\big\},
$$
for some $L\le 2^{2d\ell}$.
This finishes the proof of the lemma.
\end{proof}

 The last lemma shows that the function  $\Delta_{\bX,\bY,c}$ actually has an
  $(L,q,\lambda)$-group cover, for some parameters $L\le 2^{2d\ell}$, $q\le \ell$ and $\lambda\le \ell^{O(\ell)}$. 

\begin{lemma}\label{coverlemma2}
Let $\bX$ and $\bY$ be two distributions, each supported over at most $\ell$ strings from $\zo^n$.
For any $d\in [n]$ and $c\in \{0,1\}^d$, $\Delta_{\bX,\bY,c}$ has an $(L,q,\ell^{O(\ell)})$-group cover 
  for some $L\le 2^{2d\ell}$ and~$q\le \ell$.
\end{lemma}
\begin{proof}
First we apply Lemma \ref{coverlemma} to obtain an
  $L$-cover $\{(T_a,\calS_a):a\in [L]\}$ of $\Delta:=\Delta_{\bX,\bY,c}$~for some $L\le 2^{2d\ell}$.
It suffices to show that the $L$ positive numbers $|\Delta(T_a)|$, $a\in [L]$, 
  can be divided into at most $\ell$ groups such that 
  any two in the same group have the ratio bounded from above by $\ell^{O(\ell)}$.
  
Let $p_1,\ldots,p_{\ell'}>0$ be probabilities of strings in $\bX$ for some $\ell'\le \ell$ 
  and $q_1,\ldots,q_{\ell''}>0$ be~probabilities of strings in $\bY$ for some $\ell''\le \ell$.
The observation is that every number $|\Delta(T_a)|$ is a linear form over the $p_i$'s and $q_i$'s 
  with coefficients $-1,0$ or $1$.
This motivates the following claim:

\def\cc{\mathbf{c}} \def\AA{\mathbf{A}} \def\xx{\mathbf{x}} \def\bb{\mathbf{b}}

\begin{claim}\label{hehe}
Let $u_1,\ldots,u_g>0$ be $g$ \emph{(not necessarily distinct)} positive numbers.
Let $$V=\Big\{v>0: v= c_1u_1+\cdots+c_g u_g \ \hspace{0.05cm}\text{for some $c_1,\ldots,c_g\in \{-1,0,1\}$}\Big\}.$$
Then there cannot exist $g+1$ numbers $v_1,\ldots,v_{g+1}$ in $V$ satisfying $v_{g+1}>\cdots >v_1$
  and $$\frac{v_{i+1}}{v_i}\ge (g+2)!,\quad\text{for all $i\in [g]$.}$$
\end{claim}
\begin{proof}
Assume for a contradiction that such $g+1$ numbers $v_1,\ldots,v_{g+1}$ exist in $V$ and let
$$
v_i=c_{i,1}u_1+\cdots+c_{i,g}u_g
$$
where $c_{i,j}\in \{-1,0,1\}$ for each $i\in [g+1]$.
Given that these are $g+1$ many $g$-dimensional vectors
  $c_i$ $=(c_{i,1},\ldots,c_{i,g})$, let $i^*\leq g+1$ be the smallest integer such that 
  $c_{i^*}$ can be written as a linear combination of $c_1,\ldots,c_{i^*-1}$:
  $c_{i^*}=\alpha_1c_1+\cdots+\alpha_{i^*-1}c_{i^*-1}$, which implies that
\begin{equation}\label{hehehehe1}
v_{i^*}=\alpha_1v_1+\cdots+\alpha_{i^*-1}v_{i^*-1} 
\le |\alpha_1|\cdot v_1+\cdots+|\alpha_{i^*-1}|\cdot v_{i^*-1}.
\end{equation}
We show below  that the magnitude of coefficients $\alpha_1,\ldots,\alpha_{i^*-1}$ is relatively small, 
  which leads to a contradiction because we assumed that $v_{i^*}$ 
  is much bigger than $v_{i^*-1},\ldots,v_1$.

To see this, note that $(\alpha_1,\ldots,\alpha_{i^*-1})$ is the solution
  to a $(i^*-1)\times (i^*-1)$ linear system
$
Ax=b
$
where $A$ is a $\{-1,0,1\}$-valued $(i^*-1)\times (i^*-1)$
  full-rank matrix and $b$ is a $\{-1,0,1\}$-valued vector. 
(In more detail, one can take $A$ to be a full-rank $(i^*-1)\times (i^*-1)$ submatrix of the matrix that consists of $c_1,\ldots,c_{i^*-1}$
  as columns and  take the vector $b$ to be the corresponding entries of $c_{i^*}$.) 
It follows from Cramer's rule that each entry of $A^{-1}$ has magnitude 
  at most $(i^*-1)!$ and thus, each entry of $A^{-1}b$ has absolute value
  at most $(i^*-1)\cdot (i^*-1)!<i^*!\le (g+1)!$
This contradicts with (\ref{hehehehe1}) and the assumption
  that $v_1<\ldots<v_{i^*-1}\le v_{i^*}/(g+2)!$.
\end{proof}
Claim \ref{hehe} gives us the following procedure to partition $[L]$ into $A_1,\ldots,A_q$ 
  for some $q\le \ell$: \vspace{0.15cm}
\begin{enumerate}
\item Set $i=1$ and $\calL=[L]$.\vspace{-0.1cm}
\item While $\calL$ is nonempty do\vspace{-0.1cm}
\item \ \ \ \ Let $v$ be the smallest $|\Delta(T_a)|$, $a\in \calL$.\vspace{-0.1cm}
\item \ \ \ \ Remove from ${\cal L}$ and add to $A_i$ every $a\in \calL$ with $|\Delta(T_a)|\le (2\ell+2)!\cdot v$, and increment $i$. \vspace{ 0.15cm}
\end{enumerate}
It follows from Claim \ref{hehe} that when $\calL$ becomes empty at the end,
  the number of $A_i$'s we created can be no more than $\ell$.
Furthermore, every $a$ and $a'$ that belong to the same $A_i$ have the ratio
  of $|\Delta(T_a)|$ and $|\Delta(T_{a'})|$ bounded by $(2\ell+2)!=\ell^{O(\ell)}.$
This finishes the proof of the lemma.
\end{proof} 

\subsection{Main Algorithm}\label{algorithmsec}

\input{ThreeString.tex}

%% file: ThreeString.tex
We start with an algorithm, {based on dynamic programming,} for estimating the $k$-deck 
  of a distribution $\bX$ over $\{0,1\}^n$.  
  
	\begin{theorem} \label{thm:traces-to-deck}
		 {Let $k\in [n]$.} 
		There is an algorithm with the following~performance~guarantee:  for~any distribution $\bX$ over strings in $\zo^n$, if 
		the algorithm is given  
\begin{equation*}
M= O\left(\frac{k}{\xi^2(1-\delta)^{2k}}\right)
\end{equation*}
		many samples from $\Del_\delta(\bX)$ then {with probability at least $0.99$}
		the algorithm outputs a nonnegative $2^k$-dimensional vector $Q$ with
$ 
\|Q-\deck_k(\bX)\|_\infty\le \xi.
$ 
		Its running time is $2^kM   \cdot \poly(n)$.
	\end{theorem}
	
	\begin{proof}
Let $x^1,\ldots,x^p$ be the support of $\bX$. 
Then for each string $v\in \{0,1\}^k$, we have 
\begin{align*}
\Ex_{\bz \sim \Del_\delta(\bX)}\big[\#(v,\bz)\big] &= (1-\delta)^k \cdot
		\left(
		\bX(x^1)\cdot { {\#(v,x^1)} } + \cdots + \bX(x^p)\cdot { {\#(v,x^p)} }
		\right)\\[-1.5ex]
&=(1-\delta)^k\cdot \Ex_{\bx\sim \bX } \big[\#(v,\bx)\big] 
=(1-\delta)^k\cdot \#(v,\bX) 
=(1-\delta)^k\cdot  {n\choose k}\cdot \big(\deck_k(\bX)\big)_v.
\end{align*}
		The first equation is because for a given size-$k$ subset $S \subseteq [0:n-1]$ of indices at which $v$ matches $x^i$, all of the positions in $S$ ``survive'' into a string $\bz \sim \Del_\delta(x^i)$ with probability exactly $(1-\delta)^k.$ 

As a result, it suffices to estimate $\Ex[\#(v,\bz)]$ to additive accuracy
  $\pm \xi (1-\delta)^k {n\choose k}$ for every string $v\in \{0,1\}^k$.
%
For any fixed string $v \in \{0,1\}^k$, by a standard Chernoff bound, using 
$$M=O\left(\frac{k }{\xi^2(1-\delta)^{2k}}\right)$$
 samples the empirical estimate of  $\E[\#(v,\bz)]$ will have the desired additive $\xi (1-\delta)^k {n\choose k}$ accuracy except with failure probability $0.01/2^k$.
The success probability of 0.99 follows from union bound. 
		
The running time of the algorithm uses the following simple observation: given $z \in \zo^{n'}$~and $v \in \zo^k$, there is a $\poly(n',k)$-time {procedure} that computes $\#(v,z).$ The {procedure} works by straightforward dynamic programming:  For each $j \in [0:k-1]$ and $i \in [0:n'-1]$, the algorithm maintains a count of the number 
		$\#(v_0\ldots v_j,z_0\ldots z_i)$.	
This then implies that the running time of the overall algorithm is $M\cdot 2^k \cdot \poly(n)$. 
This finishes the proof of the lemma.	
	\end{proof}

We prove the following main technical lemma in Sections \ref{prooftechnical1} and \ref{prooftechnical2}.  Intuitively, this lemma says that if the total variation distance between $\bX$ and $\bY$ is not too small, then for a suitable (not too large) value of $k^\ast,$ the distance between the $k^\ast$-decks of $\bX$ and $\bY$ also cannot be too small. 
 \begin{lemma}\label{maintechnical}
Let $\ell$ be a positive integer with $\ell\le \log n$.
Let $\bX$ and $\bY$ be two distributions, each~supported over at most $\ell$ strings from $\zo^n$.
Then there is a positive integer 
\begin{equation}\label{defk}
k^*=\sqrt{n}\cdot (\log n)^{O(\ell)}
\end{equation}
such that
$$
 \dtv(\bX,\bY)\le \exp\left(\sqrt{n}\cdot (\log n)^{O(\ell)}
\right)\cdot 
\|\deck_{k^*}(\bX)-\deck_{k^*}(\bY)\|_\infty.
$$
\end{lemma} 

We now present our algorithm $A$ and use Lemma \ref{maintechnical} to prove Theorem \ref{thm:upper-bound-uniform}:
 
\begin{proof}[Proof of Theorem \ref{thm:upper-bound-uniform}]
 The bound (\ref{mainbound}) we aim for holds trivially when $\ell\ge \log n$.
To see this, we first notice that 
  when $\ell\ge \log n$, the 
  sample complexity bound (\ref{mainbound}) we aim for is at least
\begin{equation}\label{hela}
\frac{\text{poly}(\ell)}{\eps^2}\cdot \left(\frac{1}{1-\delta}\right)^n.
\end{equation}
With $(1/(1-\delta))^n$ samples from $\Del_\delta(\bX)$, we expect to see a full string of length $n$ where~no bits~are 
  deleted and we know that such a string is drawn directly from $\bX$.
This means that, with (\ref{hela}) many samples, we receive
  $\text{poly}(\ell)/\eps^2$ draws from $\bX$ with high probability. When the latter happens, 
  the empirical estimation $\tilde{\bX}$ of $\bX$ satisfies 
  $\dtv(\bX,\tilde{\bX})\le \eps$ with high probability.
This allows us to focus on the case when $\ell\le \log n$ in the rest of the proof
  (so Lemma \ref{maintechnical} applies).

Let $\eps$ be the total variation distance we aim for in Theorem \ref{thm:upper-bound-uniform}.
Let $k^*$ be the parameter in (\ref{defk}).
Let $\xi$ be a parameter to be specified later.
By Theorem \ref{thm:traces-to-deck}, the algorithm $A$ can first use  
\begin{equation}\label{lululu}
M^*= O\left(\frac{k^* }{\xi^2 ({1-\delta})^{2k^*}}\right)  
\end{equation}
  samples to obtain an estimate $Q$ of $\deck_{k^*}(\bX)$ such that 
\begin{equation}\label{blabla}
\|Q-\deck_{k^*}(\bX)\|_\infty\le \xi,
\end{equation}
and it succeeds in obtaining such an estimate with probability at least $0.99$.

With $Q$ in hand the algorithm $A$
computes $ \|Q-\SD_{k^*}(\bY) \|_\infty$ 
  for every distribution $\bY$ supported on at most $\ell$ strings such that
  the probability of each string in $\bY$ is an integer multiple of~$\xi/\ell$.
Finally the algorithm outputs the distribution $\bX^*$ that minimizes the distance (breaking ties arbitrarily).
  
We show that when $Q$ satisfies (\ref{blabla}),
  $\bX^*$ must be close to $\bX$.
We start with a simple observation that one can round $\bX$ to get a distribution $\bX'$ 
  in which the probability of each string is an integer multiple
  of $\xi/\ell$ and $d_{\text{TV}}(\bX,\bX')\le \xi$.
This can be done by rounding the probability of every {string except one}
  to the nearest multiple of $\xi/\ell$ and {setting the last probability as required so that the total probability is 1.}
We have
\begin{align*}
\big\|Q-\SD_{k^*}(\bX')\big\|_\infty&\le \big\|Q-\SD_{k^*}(\bX) \big\|_\infty
+ \big\|\SD_{k^*}(\bX)-\SD_{k^*}(\bX')\big\|_\infty\\[0.5ex]&\le 
\big\|Q-\SD_{k^*}(\bX) \big\|_\infty
+\dtv(\bX,\bX')\le 2\xi.
\end{align*}
By definition of $\bX^*$ and $\bX'$, we have $\|Q-\SD_{k^*}(\bX^*) \|_\infty \leq \|Q-\SD_{k^*}(\bX') \|_\infty \leq 2 \xi$.
As a result, 
$$
\big\|\SD_{k^*}(\bX)-\SD_{k^*}(\bX^*)\big\|_\infty\le \big\|Q-\SD_{k^*}(\bX^*)\big\|_\infty
  +\big\|Q-\SD_{k^*}(\bX)\big\|_\infty\le 3\xi.
$$
It follows from Lemma \ref{maintechnical}
 that
$$
d_{\text{TV}}(\bX,\bX^*)\le 
3\xi\cdot \exp\left(\sqrt{n}\cdot (\log n)^{O(\ell)}
\right).
$$
Finally we choose $\xi$ so that the RHS becomes $\eps$. The number of samples needed 
  in (\ref{lululu}) becomes 
$$
\left(\frac{1}{\eps }\right)^2\cdot
\left(\frac{2}{1-\delta}\right)^{\sqrt{n}\hspace{0.04cm}\cdot\hspace{0.04cm} (\log n)^{O(\ell)}}. 
$$
This finishes the proof of Theorem \ref{thm:upper-bound-uniform}. 
\end{proof}

We use the following two lemmas  to prove Lemma \ref{maintechnical}.
They are proved in Section  \ref{prooftechnical1}
  and \ref{prooftechnical2}.

\begin{lemma}\label{technicallemma1} 
 Let $d,q,L$ and $\lambda$ be positive integers satisfying
  $$d,q\le \log n \quad \text{and}\quad L,\lambda\le (\log n)^{O(\log n)}.$$  
Let $\smash{\Delta:{[0:n-1]\choose d}\rightarrow \R}$~be~a function
  that is not identically~zero and has~an 
  $(L,q,\lambda)$-group cover. 
Let $m=d(n-1)L^2$.
Then there exists a $d$-variate polynomial $\phi$ with degree at most $O(\sqrt{m}\cdot\log^{{4q+1}}m)$
  and $\|\phi\|_1=\exp(O(\sqrt{m}\cdot \log^{{4q+3}}m))$\ignore{\xnote{\color{red}I think it is standard
  to use $\|\phi\|_1$ to denote the sum of absolute value of coefficients of $\phi$?
  If so we can use it and add a line in preliminaries.
The same for $\|\Delta\|_\infty$ as $\max_{T} |\Delta(T)|$.}}  such that\vspace{0.05cm}
$$
\left|\sum_{0\le t_1<\cdots<t_d<n} \phi(t_1,\ldots,t_d)\cdot \Delta\big(\{t_1,\ldots,t_d\}\big)\right|\ge
\frac{\|\Delta\|_\infty}{\exp(O(\sqrt{m}\cdot \log^{{4q-1}}m))}.
$$ 
\end{lemma} 

 We note that the following lemma holds for any two distributions $\bX,\bY$ over $\zo^n$ regardless of their support size. 
\begin{lemma}\label{technicallemma2}
Let $d,k\in [n]$  with $k\ge d$.
Let $\bX,\bY$ be distributions each supported over strings from $\zo^n$.
Then for any string $c\in \{0,1\}^d$ and $d$-variate polynomial $\phi$~of~degree at most $k-d$, 
$$
\left|\sum_{0\le t_1<\cdots<t_d<n} \phi(t_1,\ldots,t_d)\cdot \Delta_{\bX,\bY,c}\big(\{t_1,\ldots,t_d\}\big)\right|\ \le \|\phi\|_1\cdot n^{O(k)}\cdot \|\deck_k(\bX)-\deck_k(\bY)\|_\infty.
$$
\end{lemma} 
 \begin{proof}[Proof of Lemma \ref{maintechnical}]
Let $\bX$ and $\bY$ be two distributions each supported over at most $\ell$ strings from $\{0,1\}^n$.
It then follows from Lemma \ref{large-TV-dist-implies-large-delta-max} and Lemma \ref{coverlemma2}
  that there exists a string $c\in \{0,1\}^d$ with $d=\lfloor \log (2\ell)\rfloor$ 
 such that $\Delta:=\Delta_{\bX,\bY,c}$ satisfies
$ \|\Delta \|_\infty \ge \dtv(\bX,\bY)/\ell^{O(\ell)}$  {and} 
  has an $(L,q,\lambda)$-group cover for some $L\le 2^{2d\ell}$, $q\le \ell$, and $\lambda = \ell^{O(\ell)}.$ 
As we assumed that $\ell\le \log n$, both $d$ and $q$ are at most $\log n$
  and $L,\lambda\le \ell^{O(\ell)}\le (\log n)^{O(\log n)}$
  (so Lemma \ref{technicallemma1} applies).

Let $m=d(n-1)L^2$ 
and   $\phi$ be the polynomial given in Lemma \ref{technicallemma1}.
Let $k^*= \deg(\phi)+d$ (we~set $k=k^*$ in Lemma \ref{technicallemma2};
  the choice of $k^*$ ensures that $\deg(\phi)\le k^*-d$ as required in Lemma \ref{technicallemma2})  
with 
$$k^*=O(\sqrt{m}\cdot \log^{4q+1} m)= \sqrt{n}\cdot  (\log n)^{O(\ell)}.
$$
Combining Lemma \ref{technicallemma1} and Lemma \ref{technicallemma2}, we have \vspace{0.15cm}
$$
\frac{\|\Delta\|_\infty}{\exp(\sqrt{n}\cdot (\log n)^{O(\ell)})}
\le 
\exp\left(\sqrt{n}\cdot (\log n)^{O(\ell)}\right)\cdot n^{\sqrt{n}\hspace{0.04cm}\cdot\hspace{0.04cm}
 (\log n)^{O(\ell)}}\cdot \|\deck_{k^*}(\bX)-\deck_{k^*}(\bY)\|_\infty.\vspace{0.15cm}
$$  
The lemma follows from the fact that $\|\Delta\|_\infty\ge \dtv(\bX,\bY)/\ell^{O(\ell)}$.\end{proof} 

\subsection{Proof of Lemma \ref{technicallemma1}}\label{prooftechnical1}  

Let $\smash{\Delta}$ be a function over $d$-subsets of $[0:n-1]$ that is not identically
  zero and has an $(L,q,\lambda)$-group cover $\{(T_a,\calS_a):a\in [L]\}$ 
  with a $q$-way partition $A_1,\ldots,A_q$ of $[L]$.
We start with a high-level description of the $d$-variate polynomial $\phi$. 

To evaluate $\phi$ on a tuple $(t_1,\ldots,t_d)$, we first project
  $(t_1,\ldots,t_d)$ onto a line along the direction of $(w_1,\ldots,w_d)$ for some 
  relatively small positive integers $w_1,\ldots,w_d$ to be specified later, and 
  then apply a univariate polynomial $f(\cdot)$ on the image of the projection.
In other words $\phi$ takes the form 
\begin{equation}\label{form}
\phi(t_1,\ldots,t_d)=f\big(w_1t_1+\cdots+w_dt_d\big)
\end{equation}
for some positive integers $w_1,\ldots,w_d\in [L^2]$.
 {We give details below.}

\subsubsection{The projection }

Let $m=d(n-1)L^2$ and let $w$ be the following function from size-$d$ subsets of $[0:n-1]$ to {$[0:m]$}: 
$$
w(T)=w_1t_1+\cdots+w_dt_d,\quad\text{where $T=\{t_1,\ldots,t_d\}$ with $t_1<\cdots<t_d$.}
$$
 {So $w$ is the projection function that maps a size-$d$ subset $T$ of $[0:n-1]$ (or equivalently, a sorted $d$-tuple of distinct values from $[0:n-1]$) to a location on the real line.} 
Claim \ref{blue1} implies that there exist $w_1,\ldots,w_d\in [L^2]$ such that 
   all anchor sets in the $L$-cover are mapped to distinct locations.

\begin{claim}\label{blue1}
If $w_1,\ldots,w_d$ are drawn independently and uniformly at random from $[L^2]$ 
  then $w(T_a)$ $\ne w(T_{a'})$ for all $a\ne a'\in [L]$ with probability at least $1/2$.
\end{claim}
\begin{proof}
Let $S=\{s_1,\ldots,s_d\}$ and $T=\{t_1,\ldots,t_d\}$ denote two size-$d$ subsets of $[0:n-1]$
  that satisfy $s_1<$ $\cdots <s_d$, $t_1<\cdots <t_d$ and $S\ne T$.
Then the probability that $w(S)=w(T)$ equals
\begin{equation} \label{hahaha1}
\Pr\big[w_1s_1+\cdots+w_ds_d=w_1t_1+\cdots+w_dt_d\big].
\end{equation}
As $(s_1,\dots,s_d)\ne (t_1,\dots,t_d)$, one of the $d$ quantities $s_i-t_i$ is nonzero; say without loss of generality $s_1 \neq t_1.$ Fixing any outcomes of random draws of $w_2,\dots,w_d$, there is a unique
   outcome of $w_1$ which would result in the equation in (\ref{hahaha1}), 
  and the probability that $w_1$ takes this particular outcome is either $1/L^2$ or zero (if it is not in $[L^2]$).
As a result, the probability in (\ref{hahaha1}) is at most $1/L^2$,
  and the claim follows from a union bound over ${L\choose 2}$ events.
\end{proof}

We fix such a tuple $w_1,\ldots,w_d\in [L^2]$ that satisfies Claim~\ref{blue1} for the rest of the proof.

\subsubsection{The univariate polynomial }
Now we move to the more difficult part of choosing the univariate polynomial $f$ in (\ref{form}).

\medskip
\noindent {\bf A useful tool.} 
A key tool for our construction of $f$ is a univariate polynomial $h$ with several useful properties described below.  Figure~\ref{fig:plot-of-h}  gives a schematic representation of the key upper bounds on $|h(b)|$ provided by item (2) in Lemma~\ref{uuuuu}.

\begin{figure}\centering
  \begin{minipage}[c]{0.36\textwidth}
    \includegraphics[width=2.5in]{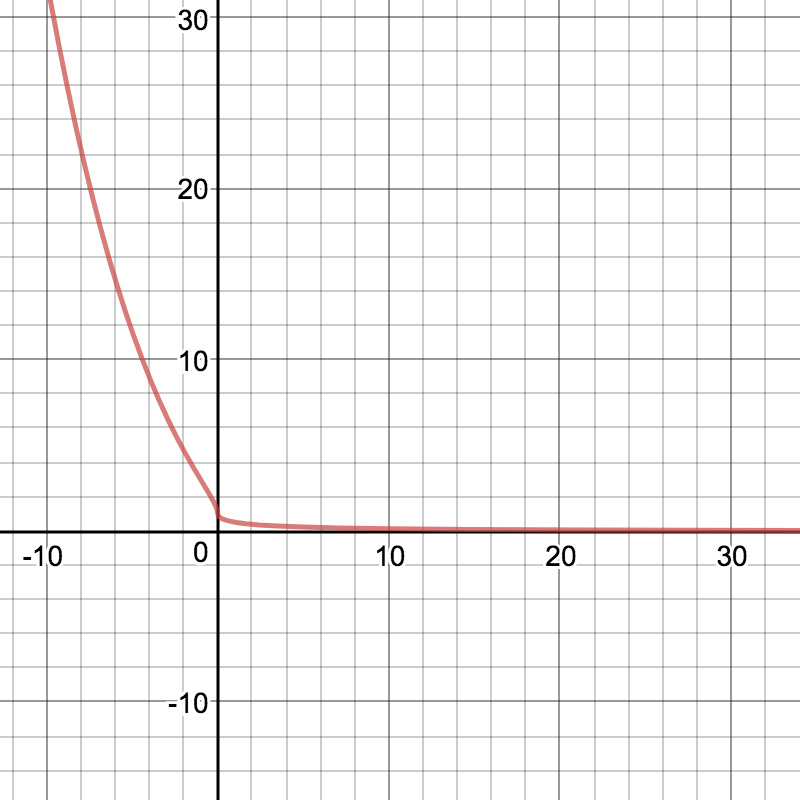}
  \end{minipage} \hskip0.5in
  \begin{minipage}[c]{0.46\textwidth}
    \caption{A schematic representation of the bounds on $|h(b)|$ given by item (2) of Lemma \ref{uuuuu}.  The three key properties are that (i) $h(0)=1$; (ii) for $b \in [m]$, the upper bound on $|h(b)|$ is very small and decreases rapidly as we move away from 0; and (iii) for $b \in [-m:-1]$, the upper bound on $|h(b)|$ is not too large and does not increase too rapidly as we move away from 0.} 
\label{fig:plot-of-h}
  \end{minipage}
\end{figure}

\begin{lemma}\label{uuuuu}
There is a univariate polynomial $h$ with the following properties:
\begin{enumerate}
\item $h$ has degree $O(\sqrt{m}\log m)$.
\item $h(0)=1$ and for each $b\in [m]$, 
$$|h(b)|\le \frac{1}{2^{\sqrt{b}}}\quad \text{and}\quad |h(-b)|\le e^{6\sqrt{b}\log m}.$$
\item  {$h$ satisfies $\|h\|_1 \leq \exp(O(\sqrt{m}\log m))$}.\vspace{0.1cm}
\end{enumerate}
\end{lemma}

 Our construction of the polynomial $h$ is based on the Chebyshev polynomial and builds on an earlier construction due to Borwein et al.~\cite{BEK99}.  We prove Lemma~\ref{uuuuu} in Appendix~\ref{ap:polynomial-h}, and we explain the role that $h$ plays in the construction of our desired univariate polynomial $f$ under the heading ``{\bf The high-level idea}'' below, after first providing some useful preliminary explanation. 

Given that our polynomial $\phi$ takes the form of (\ref{form}),
 {the crucial quantity whose magnitude we are trying to lower bound, namely}
$$
\sum_{0\le t_1<\cdots<t_d<n} \phi(t_1,\ldots,t_d)\cdot \Delta\big(\{t_1,\ldots,t_d\}\big)
$$
 {(recall the LHS of Lemma~\ref{technicallemma1}),}
can be written as

\begin{equation} \label{eq:turnip}
\sum_{b\in [0:m]} f(b)\cdot \Gamma(b),
\end{equation}
where 
$\Gamma:[0:m]\rightarrow \R$ is a function that is defined using $\Delta$ as follows: 
\begin{equation}\label{heheeq1}
\Gamma(b)=\sum_{T:\hspace{0.06cm}w(T)=b} \Delta(T),
\end{equation}
where the sum is over all $d$-subsets $T$ of $[0:n-1]$.



 To better understand $\Gamma$, we use the $(L,q,\lambda)$-group cover of $\Delta$ to introduce two new 
  sequences $\tau_0,\ldots,\tau_r$ and $m_0, \ldots, m_r$, for some value $r \in [0:q-1]$ that is defined below. 
We start with some notation. For each $i\in [q]$,
  we let $\calG_i=\cup_{a\in A_i} \calS_a$ and refer to $\calG_i$ as \emph{group $i$}.
We refer to
  the $T_a$~with the smallest $w(T_a)$ among all $a\in A_i$
  as the \emph{anchor} of group $i$ and denote it by $V_i$.
(By Claim~\ref{blue1}, each group has a unique anchor and we have 
  $w(T)>w(V_i)$ for all $T\in \calG_i$ other than $V_i$.)
We let $v_i=|\Delta(V_i)|$ and $\kappa_i=w(V_i)$,
 {so $\kappa_i$ is the location that the anchor $V_i$ of $\calG_i$ is projected to.}
By the definition of an $(L,q,\lambda)$-group cover 
  and Claim \ref{blue1}, we have that each $v_i>0$, the $\kappa_i$'s are distinct,  
$$
\max_{i\in [q]} v_i\ge \frac{\|\Delta\|_\infty}{\lambda}.
$$


  Now we are ready to define $r$ and the two sequences. See Figure~\ref{fig:step-sequences} for an illustration of these sequences. 
First we let $\tau_0=\max_{i\in [q]} v_i$ and also let  
  $m_0\in [0:m] $ 
  denote the smallest $\kappa_i$ (among~all groups $i\in [q]$) with $v_i=\tau_0$.
We are done {and the value of $r$ is 0} if no $\kappa_i$ is smaller than $m_0$~{(i.e. the anchor of every other group is projected to a larger location value than $m_0$)}; otherwise, we let
  $\tau_1<\tau_0$ be the largest value of $v_i$ {over those $i\in [q]$ that have $\kappa_i < m_0$} and also let 
  $m_1<m_0$ be the smallest $\kappa_i$ such that $v_i=\tau_1$.
We~are done and the value of $r$ is 1 if no $\kappa_i$ is smaller than $m_1$  {(i.e. every other group anchor   
  is projected to a larger location value than $m_1$)}; otherwise we repeat the process. Continuing in this way, at the end we obtain two sequences:
  $$0<\tau_r<\cdots<\tau_0\qquad \text{with $\tau_0=\max_{i\in [q]} v_i\ge \frac{\|\Delta\|_\infty}{\lambda}$ and}\qquad 0\le m_r<\cdots<m_0 \leq m,$$ for some value $r \in [0:q-1]$. 
 {We say that $(\tau_0,\dots,\tau_r)$ is the \emph{$\tau$-step-sequence} and that $(m_0,\dots,m_r)$ is the \emph{$m$-step-sequence} for $\Gamma$.

\begin{figure}
\centering
    \includegraphics[width=6.3 in]{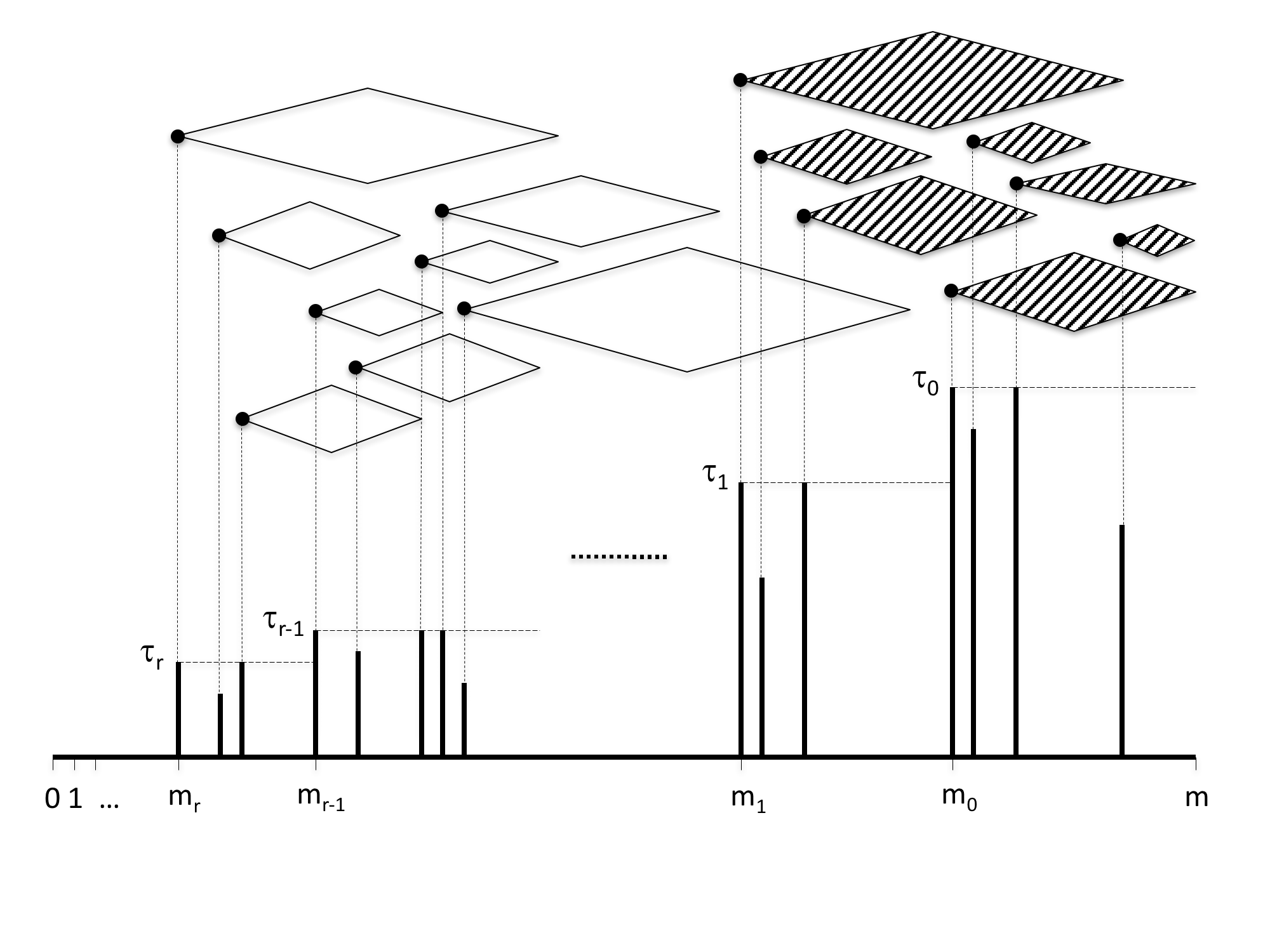}\vspace{-1.4cm}
\caption{{An illustration of a $\tau$-step-sequence and its associated $m$-step-sequence. The values $\tau_r < \tau_{r-1} < \cdots < \tau_1 < \tau_0$ (which may be arbitrary real positive values) are the heights of the bars at locations $0 \le m_r <m_{r-1}< \cdots < m_1 < m_0 \leq m$ (these locations are integers).  The location of each vertical bar corresponds to some $\kappa_i, i \in [q]$, and its height is $v_i$; {the corresponding group
  $\calG_i$ is illustrated as a diamond, with $V_i$ being its left corner}. Note that all the bars between locations $m_i$ and $m_{i-1}$ have height at most $\tau_i$.} {See Example~\ref{ex} for an explanation of why certain diamonds are shaded in the figure.}} 
\label{fig:step-sequences}
\end{figure}
  
  \medskip

\noindent {\bf The high-level idea.} Before entering into further details we give intuition for the polynomial $f$.  Looking ahead to (\ref{eq:fdef}), the polynomial $f$ is essentially a translation of the polynomial $h$ depicted in Figure~\ref{fig:plot-of-h}, i.e. \hspace{-0.03cm}$f(x)$ is essentially $h(x-m_\alpha)$ for some $\alpha\in [0:r]$.\footnote{The exponent of $h$ in the exact definition of our $f$ given in Equation~(\ref{eq:fdef}) is needed for technical reasons that are not important for this intuitive explanation.}  Recalling the key properties~of $h$, we see that 
\begin{itemize}
\item $f(m_\alpha)=1.$\vspace{-0.06cm}
\item $|f(b)|$ is ``very small'' for $b > m_\alpha$; and\vspace{-0.06cm}
\item  $|f(b)|$ is ``not too large'' for $b < m_\alpha$.
\end{itemize}

The crux of our analysis below is to establish that there is a suitable value $m_\alpha$ in the $m$-step-sequence which is such that the magnitude of the single summand
$f(m_\alpha) \cdot \Gamma(m_\alpha)$ in (\ref{eq:turnip}) {is greater than} the contribution of all other summands  in (\ref{eq:turnip}). 

To gain some intuition for why this is the case, let us pretend that instead of the $\Gamma$ being~defined as in (\ref{eq:turnip}), the definition of $\Gamma$ instead only took a sum over the $q$ anchors  $V_1,\dots,V_q$ of the $q$ groups $\calG_1,\ldots,\calG_q$
  (i.e. $\Gamma$ is supported on $\kappa_i$, $i\in [q]$, with $|\Gamma(\kappa_i)|=v_i$). 
   Of~course this is not actually~the case since each group $\calG_i$~in general contains
   many more sets than just its anchor $T_i$, but it turns out that the effect of other sets in 
     $\supp(\Delta)$ will only cost us some extra $n^d\lambda$ factors in the analysis (corresponding to the $n^d\lambda$ factors in properties (ii) and (iii) of $\Gamma_0,\dots,\Gamma_r$ as described below,
where $\smash{n^d\ge {n\choose d}}$ just serves as a  bound for the number of size-$d$ subsets) which turn  out to be manageable.  

In this hypothetical scenario the only nonzero values of $\Gamma(b)$ that would enter the picture would be the $v_i$ values at locations $\kappa_i$, $i\in [q]$, 
which are the heights of the bars in Figure~\ref{fig:step-sequences}.
The desired $m_\alpha$ could then be identified as follows:  
\begin{flushleft}\begin{itemize}
\item We proceed in an inductive fashion. For each $p\in [0:r]$, we show that there is a choice of 
  $\alpha\in [0:p]$ such that, by setting $f(x)=h(x-m_\alpha)$, the value of $|f(m_\alpha)\cdot \Gamma(m_\alpha)|$
  outweighs $|f(b)\cdot \Gamma(b)|$ for every other $b\in  [m_p:m]$.
The choice of $\alpha$ at the end of the induction when $p$ reaches $r$ gives us the desired location
$m_\alpha$ for the translation of $h$
  to define $f$.

\item The base case when $p=0$ is trivial by setting $\alpha=0$ and $f(x)=h(x-m_0)$. 
Here we have that $|f(m_0)\cdot \Gamma(m_0)|$ outweighs $|f(b)\cdot \Gamma(b)|$ 
  for all $b>m_0$ because $|\Gamma(m_0)|=\tau_0\ge |\Gamma(b)|$ by the definition of
  our step-sequences and the fact that $f(m_0)=1$ is ``much larger'' than
  $|f(b)|$ for $b>m_0$.

\item 
Next we move to $p=1$, and now we need to take $\Gamma(b)$, $b\in [m_1:m_0-1]$, into consideration.
To this end we compare $\tau_0/\tau_1$ with $\exp(\sqrt{m_0-m_1})$ and consider the following two cases.
\begin{itemize}
\item If $\tau_0/\tau_1$ is larger then we can keep $\alpha=0$ and $f(x)=h(x-m_0)$ because
  $|f(m_0)\cdot \Gamma(m_0)|$ outweighs $|f(m_1)\cdot \Gamma(m_1)|$
  (since $\Gamma(m_1)=\tau_1$ and $f(m_1)$ is roughly\footnote{
  	This is not entirely precise because 
    in (2) of Lemma \ref{uuuuu} there is indeed an extra factor of $\log m$ in the exponent on the left side of $0$; overcoming this factor of $\log m$ is the reason why we end up with the exponent as in  (\ref{eq:fdef}).} 
$\exp(\sqrt{m_0-m_1})$)
  as well as $|f(b)\cdot \Gamma(b)|$ for all $b\in [m_1+1:m_0-1]$ 
  (since $|f(m_1)|>|f(b)|$ and by the definition of our step-sequences, $|\Gamma(m_1)|\ge |\Gamma(b)|$).
By the inductive hypothesis we also know that $|f(m_0)\cdot \Gamma(m_0)|$ 
   outweighs $|f(b)\cdot \Gamma(b)|$ for all $b>m_0$.\vspace{0.15cm}

\item Otherwise (if $\tau_1$ is larger than $\tau_0/\exp(\sqrt{m_0-m_1})$)
  we show that setting $\alpha=1$ and $f(x)=h(x-m_1)$ works.
On the one hand, $|f(m_1)\cdot \Gamma(m_1)|$ outweighs $|f(b)\cdot \Gamma(b)|$
  for $b\in [m_1+1:m_0-1]$ since $|\Gamma(b)|\le |\Gamma(m_1)|$ by the definition of
  our step-sequences and the fact that $f(m_1)=1$ is ``much larger'' than $|f(b)|$ (similar to the base case).
On the other hand, $|f(m_1)\cdot \Gamma(m_1)|=\tau_1$ outweighs $|f(m_0)\cdot \Gamma(m_0)|=
 |f(m_0)|\cdot \tau_0$ (since
 $|f(m_0)|$ is, roughly speaking, $\exp(-\sqrt{m_0-m_1})$) as well as $|f(b)\cdot \Gamma(b)|$
 for $b>m_0$ (since $|f(m_0)|>|f(b)|$ and $|\Gamma(m_0)|\ge |\Gamma(b)|$ so the contribution from $b$
 is smaller than that from $m_0$).
\end{itemize}
\item Continuing in this fashion, we show that, if $\alpha$ is the choice for some $p\in [0:r-1]$,
  then for $p+1$ we can either keep the same choice of $\alpha$ or move $\alpha$ to $p+1$,
  depending on the result of a similar comparison between $\tau_\alpha/\tau_{p+1}$ and $\exp(\sqrt{m_\alpha-m_{p+1}})$.
This finishes the induction.
\end{itemize}\end{flushleft}
The above reasoning is formalized in the statement and proof of the (crucial) Lemma~\ref{starlem}, which additionally has to deal with the complication that it must address the real scenario rather than the hypothetical simplification considered in the informal description above.

\vspace{0.1cm} 

\medskip

 Now we turn to the details.
For each $b\in [0:m]$, we let
$$
\calG_{\ge b}=\bigcup_{i\in [q]:\hspace{0.06cm} \kappa_i\ge b} \calG_i.
$$
For each $p\in [0:r]$ let $\Gamma_p$  denote the following function on $[0:m]$:
$$
\Gamma_p(b)=\sum_{T\in \calG_{\ge m_p}:\hspace{0.05cm}w(T)=b} \Delta(T).
$$
In words, the value of $\Gamma_p$ evaluated at a location value $b$ is obtained as follows:  for each group $\calG_i$  for which the location $\kappa_i$ that the anchor set $V_i$ is projected to is at least $m_p$, 
we sum the value of $\Delta(T)$ over all $T \in \calG_i$ which are mapped by the projection function $w$ to the location $b$.
\begin{example}
\label{ex}In Figure \ref{fig:step-sequences}, only $\calG_i$'s that correspond to shaded diamonds are considered in $\Gamma_1$. 
  \end{example}

We have the following properties from our choices of $\tau_i$'s and $m_i$'s:
\begin{flushleft}

\begin{enumerate}
\item[(i)] $\Gamma=\Gamma_r$. {This is because every location $\kappa_i$ is at least $m_r$.}\vspace{-0.06cm}
\item[(ii)] $\Gamma_0$ is such that $\Gamma_0(b)=0$ for all $b<m_0$,
  $|\Gamma_0(m_0)|=\tau_0$, and $$|\Gamma_0(b)|\le n^d \lambda\tau_0$$ for all $b>m_0$.
  (The last bound holds just because there are at most $n^d$ many size-$d$ subsets {and the maximum value of $|\Delta(T)|$ on any $T$ contributing to the sum $\Gamma_0(b)$ is at most $\lambda\tau_0$.})\vspace{-0.06cm}
\item[(iii)] {Generalizing the previous property,} for each $p\in [r]$,  $\Gamma_p(b)=0$ for $b<m_p$,
  $|\Gamma_p(m_p)|=\tau_p$ and $$ \big| \Gamma_p(b)-\Gamma_{p-1}(b) \big|
  \le n^d\lambda\tau_p$$ {for all $b>m_p$}
  {(since the maximum magnitude of $\Delta(T)$ on any subset $T$ contributing to the sum $\Gamma_p(b)$ but not to $\Gamma_{p-1}(b)$ is at most $\lambda\tau_p$)}.
 
   \end{enumerate}\end{flushleft}

We prove the following crucial lemma, Lemma~\ref{starlem}, concerning $\Gamma_p$ by induction on $p$.
Intuitively, the lemma states that for each $p$ there is a suitable index $\alpha \leq p$  (so $m_\alpha \geq m_p$) such that (a) the magnitude~of $\Gamma_p(m_\alpha)$ is not too small compared to $\tau_0$ (this is given by (\ref{hueq1})); (b) for locations $b > m_\alpha$ the~magnitude of $\Gamma_p(b)$ is not too large compared to the magnitude of $\Gamma_p(m_\alpha)$ (this is given by (\ref{hueq2})); and (c) for locations $b$ between $m_p$ and $m_\alpha$ the magnitude of $\Gamma_p(b)$ is small compared to the magnitude of $\Gamma_p(m_\alpha)$ (this is given by (\ref{hueq3})).
In all three places the meaning of ``small''  or ``large'' is specified~by a second
  parameter $\beta$ which can grow slowly with $p$. }

 \begin{lemma}\label{starlem}
Assume that $d \le \log n$ and $\lambda\le (\log n)^{O(\log n)}$.
Then for each $p\in [0:r]$ there are two parameters $\alpha_p\in [0:p]$ and $\beta_p\in [0:4p+3]$
  (letting~$\alpha$ denote $\alpha_p$ and $\beta$ denote $\beta_p$ below for convenience) such that
\begin{equation}\label{hueq1}
\big|\Gamma_p(m_\alpha)\big|\cdot 2^p\cdot  \exp\big(\sqrt{m}\cdot \log^{\beta} m\big)\ge \tau_0
\end{equation} 
and every index $b\in [m_p:m]$ satisfies 
\begin{enumerate}
\item If $b\ge m_{\alpha }$, then 
\begin{equation}\label{hueq2}
\big|\Gamma_{p}(b)\big|\le 
\big|\Gamma_{p}(m_\alpha)\big|\cdot 2^p\cdot  \exp\big({ \sqrt{b-m_\alpha}\cdot\log^\beta m}\big);
\end{equation} 
\item If $m_p\le b<m_{\alpha}$, then 
\begin{equation}\label{hueq3}
\big|\Gamma_p(b)\big| \cdot \exp\big({\sqrt{ m_\alpha-b}\cdot\log^{\beta+3} m}\big)\le 
 2^p\cdot \big|\Gamma_p(m_\alpha)\big|.
\end{equation}
\end{enumerate}
\end{lemma}
\begin{proof}
The base case when $p=0$ follows from properties of $\Gamma_0$ by setting $\alpha_0=0$ and $\beta_0=3$ since $d\le  \log n$ and $\lambda\le (\log n)^{O(\log n)}$ imply that $n^d\lambda\ll \exp(\log^3 m)$.
  
For the induction step, we assume that the claim holds for some $p\in [0:r-1]$ with
  parameters $\alpha=\alpha_p\in [0:p]$ and $\beta=\beta_p\in [0:4p+3]$, 
  and now we prove it for $p+1$.
We consider two cases:
\begin{align}\label{IH2}
\tau_{p+1}\cdot 2n^d\lambda\cdot \exp\big({\sqrt{m_\alpha-m_{p+1}}\cdot\log^{\beta+3} m}\big)&\le 
\big|\Gamma_p (m_\alpha)\big| \quad\text{or}\\[0.7ex] \label{IH}
\tau_{p+1}\cdot  2n^d\lambda \cdot \exp\big({\sqrt{m_\alpha-m_{p+1}}\cdot\log^{\beta+3} m}\big)&> 
\big|\Gamma_p (m_\alpha)\big|.
\end{align}

For the case of (\ref{IH2}),  
  we verify that the claim  holds for $p+1$ by setting $\alpha_{p+1}=\alpha$ and $\beta_{p+1}=\beta$.
First, {by property (iii) and (\ref{IH2}) we have}
\begin{equation}\label{heheuse1}
\big|\Gamma_{p+1}(m_\alpha)\big|\ge \big|\Gamma_p(m_\alpha)\big|
  -n^d\lambda\tau_{p+1}\ge (1-o_n(1))\cdot \big|\Gamma_p(m_\alpha)\big|.
\end{equation}
As a result, (\ref{hueq1}) holds for $p+1$ since
\begin{align*}
\big|\Gamma_{p+1}(m_\alpha)\big|&\cdot 2^{p+1} \cdot \exp\big(\sqrt{m}\cdot \log^{\beta}m\big)
\\[0.5ex] &\ge (1-o_n(1))\cdot \big |\Gamma_p(m_\alpha)\big|\cdot 2^{p+1}\cdot 
\exp\big(\sqrt{m}\cdot \log^{\beta}m\big)>\tau_0,
\end{align*}
by the inductive hypothesis of (\ref{hueq1}) for $p$.
For each $b\ge m_\alpha$, compared with the inductive hypothesis (\ref{hueq2}) for $p$,
  the LHS of (\ref{hueq2}) for $p+1$ goes up (additively) by at most $n^d\lambda\tau_{p+1}$ (by
  property (iii)),~but it follows from~(\ref{heheuse1}) that the RHS of (\ref{hueq2}) for $p+1$
  goes up multiplicatively by a factor of $2(1-o_n(1))$.
Since {by (\ref{IH2}) we have} $|\Gamma_p(m_\alpha)|\gg n^d\lambda\tau_{p+1}$, (\ref{hueq2}) also holds for $p+1$.
Similarly for each $m_p\le b<m_\alpha$, compared to the inductive hypothesis (\ref{hueq3}) for $p$,
 {again by property (iii)}
  the LHS of (\ref{hueq3}) for $p+1$ goes up (additively) by at most
\begin{align*}
\tau_{p+1}\cdot n^d\lambda\cdot \exp\big(\sqrt{m_\alpha-b}\cdot \log^{\beta+3} m\big)
&<\tau_{p+1}\cdot n^d\lambda\cdot \exp\big(\sqrt{m_\alpha-m_{p+1}}\cdot \log^{\beta+3} m\big) 
 \\[0.8ex] &\le \big|\Gamma_p(m_\alpha)\big|\big/2, \tag{by (\ref{IH2})}
\end{align*}
while the RHS goes up (additively) by at least
$$\big(2^{p+1}(1-o_n(1))-2^p\big)\cdot \big|\Gamma_p(m_\alpha)\big|\ge \big|\Gamma_p(m_\alpha)\big|\big/2. 
$$
So (\ref{hueq3}) still holds for $p+1$.
Finally, for each $m_{p+1}\le b<m_p$, we have {by property (iii) that}
\begin{align*}
\big|\Gamma_{p+1}(b)\big| \cdot \exp\big({\sqrt{ m_\alpha-b}\cdot\log^{\beta+3} m}\big)
&\le  \tau_{p+1}\cdot n^d\lambda\cdot \exp\big({\sqrt{ m_\alpha-m_{p+1}}\cdot\log^{\beta+3} m}\big)\\[0.5ex]
&\le  \big|\Gamma_p(m_\alpha)\big|\big/2  \tag{by {(\ref{IH2})}}\\[0.5ex]
&\le 2^{p+1}\cdot \big|\Gamma_{p+1}(m_\alpha)\big|. 
\tag{by {(\ref{heheuse1})}}
\end{align*} 
So (\ref{hueq3}) also holds for such ${b}$'s.
This finishes the induction step assuming (\ref{IH2}).

We now consider (\ref{IH}) and show
  that the claim holds for $p+1$ by setting $\alpha_{p+1}=p+1$ and~$\beta_{p+1}$ $=\beta+4\le 4p+7$.
  First we prove that (\ref{hueq1}) holds for $p+1$.
Combining (\ref{IH}), the inductive hypothesis (\ref{hueq1}) for $p$, and that $|\Gamma_{p+1}(m_{p+1})|=\tau_{p+1}$, we have 
\begin{align*}
\tau_0&\le \big|\Gamma_p(m_\alpha)\big|\cdot 2^{ p} \cdot\exp\big(\sqrt{m}\cdot \log^{\beta }m\big)\\[0.5ex]
&<\tau_{p+1}\cdot 2^{p+1}n^{d }\lambda\cdot \exp\big(\sqrt{m}\cdot \log^{\beta }m+\sqrt{m}\cdot\log^{\beta+3} m\big)\\[0.5ex]
&<\big|\Gamma_{p+1}(m_{p+1})\big|\cdot 2^{p+1} \cdot \exp\big(\sqrt{m}\cdot \log^{\beta+4} m\big),
\end{align*}
where the last inequality used the assumptions on $d$ and $\lambda$ so that $n^d\lambda\ll \exp(\log^3 m)$.
Next, we consider each $b\ge m_{p+1}$ and show that (\ref{hueq2}) holds for $p+1$ {(note that the conditions for (\ref{hueq3}) never occur in this case since~we set $\alpha_{p+1}=p+1$ and hence $m_{\alpha_{p+1}}=m_{p+1}$)}:
\begin{align*}
\big|\Gamma_{p+1}(b)|&\le \big|\Gamma_{p+1}(m_{p+1})\big|\cdot 2^{p+1}\cdot 
  \exp\big(\sqrt{b-m_{p+1}}\cdot \log^{\beta+4} m\big)\\[0.5ex]
&=\tau_{p+1}\cdot 2^{ p+1 } \cdot 
  \exp\big(\sqrt{b-m_{p+1}}\cdot \log^{\beta+4} m\big).
\end{align*}
This is trivial for $m_{p+1}\le b<m_p$ as $ |\Gamma_{p+1}(b) |\le n^d\lambda \tau_{p+1}$ {by property (iii)} 
  and $n^d\lambda\ll \exp(\log^3 m)$. 
For each $b:m_p\le b<m_{\alpha}$, we have by the inductive hypothesis (\ref{hueq3}) that
\begin{equation}\label{IH1}
\big|\Gamma_{p}(b)\big|\cdot \exp\big(\sqrt{m_\alpha-b}\cdot \log^{\beta+3} m\big)\le  
2^p\cdot \big|\Gamma_p(m_\alpha)\big|. 
\end{equation}
Combining (\ref{IH}) with (\ref{IH1}), we have 
\begin{align*}
\big|\Gamma_{p+1}(b)\big|&\cdot  \exp\big(\sqrt{m_\alpha-b}\cdot \log^{\beta+3} m\big)\\[0.5ex]
&\le \left(\big|\Gamma_{p}({b})\big| +n^d\lambda\tau_{p+1}\right) \cdot  \exp\big(\sqrt{m_\alpha-b}\cdot \log^{\beta+3} m\big) \tag{{by property (iii)}}\\[0.5ex]
&\le  2^p\cdot \big|\Gamma_p(m_\alpha)\big|
+n^d\lambda\tau_{p+1}\cdot   \exp\big(\sqrt{m_\alpha-m_{p+1}}\cdot \log^{\beta+3} m\big)
\tag{{by (\ref{IH1})}}\\[0.8ex]
&\le \tau_{p+1}\cdot   2^{p+2}n^d\lambda \cdot \exp\big(\sqrt{m_\alpha-m_{p+1}}\cdot \log^{\beta+3} m\big). \tag{{by (\ref{IH})}}
\end{align*}
Using $\sqrt{m_\alpha-m_{p+1}}-\sqrt{m_\alpha-b}\le \sqrt{b-m_{p+1}}$, we have 
\begin{align*}
\big|\Gamma_{p+1}(b)\big| &\le\tau_{p+1}\cdot 2^{p+2}n^d\lambda \cdot  \exp\big(\sqrt{b-m_{p+1}}\cdot \log^{\beta+3}m\big)\\[0.6ex]
&<\tau_{p+1}\cdot 2^{ p+1 } \cdot  \exp\big(\sqrt{b-m_{p+1}}\cdot \log^{\beta+4}m\big),
\end{align*}
using $n^d\lambda\ll \exp(\log^3 m)$.
The last case for (\ref{hueq2}) is that $b \geq m_\alpha$. By the inductive hypothesis,
\begin{equation}\label{fff}
\big|\Gamma_{p}(b)\big| \le \big|\Gamma_p(m_\alpha)\big| \cdot 2^{ p }\cdot
 \exp\big(\sqrt{b-m_\alpha}\cdot \log^\beta m\big).
\end{equation}
Combining (\ref{IH}) with (\ref{fff}), we have
\begin{align*}
\big|\Gamma_{p+1}(b)\big|&\le 
\big|\Gamma_p(b)\big|+n^d\lambda\tau_{p+1} \tag{{by property (iii)}}\\[0.2ex]
&\le
\tau_{p+1}\cdot 2^{p+1}n^{d}\lambda\cdot \exp\big(\sqrt{b-m_\alpha}\cdot \log^\beta m+\sqrt{m_\alpha-m_{p+1}}\cdot \log^{\beta+3} m\big)+n^d\lambda\tau_{p+1}.
\end{align*}
Using $\sqrt{b-m_\alpha}+\sqrt{m_\alpha-m_{p+1}}\le 2\sqrt{b-m_{p+1}}$ and $n^d\lambda\ll\exp(\log^3 m)$, 
 we have
$$
\big|\Gamma_{p+1}(b)\big|
<\tau_{p+1}\cdot 2^{ p+1 } \cdot \exp\big(\sqrt{b-m_{p+1}}\cdot \log^{\beta+4} m\big).
$$
This finishes the proof of the {lemma}.
\end{proof}

Finally we combine Lemma \ref{uuuuu} and Lemma \ref{starlem} to prove Lemma \ref{technicallemma1}.

\begin{proof}[Proof of Lemma \ref{technicallemma1}]
Recall that $d,q\le \log n $ and $\lambda,L\le (\log n)^{O(\log n)}$.

Let $\alpha\in [0:r]$ and $\beta\in [0:4r+3]$ be the final parameters 
  that satisfy Lemma \ref{starlem}~for $\Gamma_r=\Gamma$.
We define the polynomial $f$ using $h$ from Lemma \ref{uuuuu} as follows
\begin{equation}
f(x)=\Big(h(x-m_{\alpha})\Big)^{\big\lceil 3\log^{\beta+1}m\big\rceil}. \label{eq:fdef}
\end{equation}
It follows from Lemma \ref{uuuuu} that
  $f$ has degree $$\text{deg}(f)=
  O\big(\sqrt{m}\cdot \log^{\beta+2} m \big)=
  O\big(\sqrt{m}\cdot \log^{4q+1} m \big).$$
 using $r<q$ and $\beta\le 4r+3$. 
Moreover, we have
\begin{align*}
\|f\|_1\le \exp\big(O(&\sqrt{m}\cdot  \log^{4q+1} m )\big)\cdot (m+1)^{\text{deg}(f)} =
   \exp\big(O(\sqrt{m}\cdot \log^{4q+2} m )\big). 
\end{align*}
It follows from the definition of $\phi$ in (\ref{form}) and 
  $w_1,\ldots,w_d\in [L^2]$ that the same degree upper bound holds for
  $\phi$ and  \begin{align*}
\|\phi\|_1\le \exp\big(O(\sqrt{m}&\cdot \log^{4q+2} m )\big)\cdot (dL^2)^{\deg(f)}
\le \exp\big(O (\sqrt{m}\cdot \log^{4q+3} m )\big).
\end{align*}
  
To analyze $\sum_b f(b)\cdot \Gamma(b)$, we show that 
  $| f(b)\cdot \Gamma(b)|\le |\Gamma(m_\alpha)|/(2m)$ for all $b\ne m_\alpha$ and thus,
$$
\left|\hspace{0.03cm}\sum_{b\in [0:m ]} f(b)\cdot \Gamma(b)\hspace{0.03cm}\right|
\ge \frac{|\Gamma(m_{\alpha})|}{2} \stackrel{(\ref{hueq1})}{\ge} \frac{\tau_0}{2^q\cdot \exp(\sqrt{m}\cdot \log^\beta m)}
\ge \frac{\|\Delta\|_\infty}{\exp(O(\sqrt{m}\cdot \log^{4q-1}m))} 
$$
using $\tau_0\ge \|\Delta\|_\infty/\lambda$.
For each $b>m_{\alpha}$, by Lemma \ref{starlem} and Lemma \ref{uuuuu} (and $q\le  \log n$) 
\begin{align*}
\big| f(b)\cdot \Gamma(b)\big| \le \big|\Gamma(m_\alpha)\big|\cdot 2^{q}\cdot  \exp\big(\sqrt{b-m_\alpha}\cdot \log^\beta m\big)
\cdot \frac{1}{2^{\sqrt{b-m_\alpha}\hspace{0.03cm}\cdot\hspace{0.03cm} 3\log^{\beta+1}m}} 
 \le \frac{|\Gamma(m_\alpha)|}{2m}.
\end{align*}
For each $b<m_{\alpha}$ we have from Lemma \ref{starlem} and Lemma \ref{uuuuu} that
$$
\big| f(b)\cdot \Gamma(b)\big|\le \frac{ 2^q\cdot |\Gamma(m_\alpha)|}{\exp(\sqrt{m_\alpha-b}\cdot \log^{\beta+3} m)}
\cdot \exp\big(O(\sqrt{m_{\alpha}-b}\cdot \log^{\beta+2} m \big)\le \frac{|\Gamma(m_\alpha)|}{2m}.
$$
This finishes the proof of Lemma~\ref{technicallemma1}.
\end{proof} 

\subsection{Proof of Lemma \ref{technicallemma2}}\label{prooftechnical2}

Let $\bX$ and $\bY$ be two distributions each supported over strings from $\zo^n$.

Given $0\le j_1<\cdots<j_d\le k-1$, we use $g_{j_1,\ldots,j_d}$ to denote the 
  following $d$-variate polynomial, \begin{equation} \label{eq:gbasis}
g_{j_1,\dots,j_d}(t_1,\dots,t_d) := {t_1 \choose j_1} \cdot {t_2-t_1-1 \choose j_2-j_1-1} \cdots 
 {t_d - t_{d-1} - 1 \choose j_d - j_{d-1} - 1} \cdot
  {n-t_d-1 \choose k-j_d-1}.
\end{equation}
To see the relevance of this polynomial to the $k$-deck, we note that given any $0\le t_1<\cdots<t_d<n$  the quantity
$g_{j_1,\dots,j_d}(t_1,\dots,t_d)$ is the number of ways to 
  pick $k$ indices from $[0:n-1]$ such that each $t_i$ is the $(j_i+1)$th smallest index picked.  

We first show that the following sum
\begin{equation}\label{sum}
\sum_{0\le t_1<\cdots<t_d<n} g_{j_1,\ldots,j_d}(t_1,\ldots,t_d)\cdot \proj\big(\bX,\{t_1,\ldots,t_d\},c\big)
\end{equation}
can be written as a low-weight linear combination of entries of $\deck_k(\bX)$.

\begin{lemma}\label{firsteasylemma}
For any $0\le j_1<\cdots<j_d\le k-1$ and any $c\in \{0,1\}^d$,
  the sum (\ref{sum}) can be written as a linear combination of entries of $\deck_k(\bX)$
  in which each coefficient is either $0$ or ${n\choose k}$.
\end{lemma} 
\begin{proof}
{Recalling the combinatorial interpretation of $g_{j_1,\dots,j_d}(t_1,\dots,t_d)$ given  after (\ref{eq:gbasis}),} we see that if we divide the sum in (\ref{sum}) by
  ${n\choose k}$, the result is precisely  
  the probability that $(\bz_{j_1},\ldots,\bz_{j_d})=c$ when
  we draw $\bx\sim\bX$, {draw} a size-$k$ subset $\bT$ of $[0:n-1]$ uniformly at random, and then set $\bz=\bx_\bT$.
The latter probability can also be expressed using entries of $\deck_k(\bX)$ as
$$
\sum_{\substack{z\in \{0,1\}^k \\ (z_{j_1},\ldots,z_{j_d})=c}} \big(\deck_k(\bX)\big)_z,
$$
as $(\deck_k(\bX))_z$ is the probability of $\bx_\bT=z$ with $\bx$ and $\bT$ drawn as above.
This finishes the proof.
\end{proof}

Next we show that, for every monomial $t_1^{r_1}\cdots t_d^{r_d}$ of degree $r_1+\cdots+r_d\le k-d$,
  there exists~a low-weight 
  linear combination of polynomials $g_{j_1,\ldots,j_d}$ that agrees with 
  $t_1^{r_1}\cdots t_d^{r_d}$ over $t_1,\ldots,t_d$ that satisfy
  $0\le t_1<\cdots<t_d<n$.
  
\begin{lemma}\label{mainhardlemma}
For any nonnegative integers $r_1,\ldots,r_d$ with $r_1 + \cdots + r_d \leq k-d$, we have that
\[
t_1^{r_1}  \cdots t_d^{r_d} = \sum_{0\le j_1 <  \cdots < j_d<k} w_{j_1,\dots,j_d} \cdot g_{j_1,\dots,j_d}(t_1,\dots,t_d),\quad\text{for all $0\le t_1<\cdots<t_d<n$,}
\]
where the coefficients $w_{j_1,\dots,,j_d}$ satisfy
	$\sum |w_{j_1,\dots,j_d}| \leq k^{O(k )}$.
\end{lemma}

Before proving Lemma \ref{mainhardlemma}, we use Lemma \ref{firsteasylemma}
  and Lemma \ref{mainhardlemma} to prove Lemma \ref{technicallemma2}.

\begin{proof}[Proof of Lemma \ref{technicallemma2}]
Combining Lemma \ref{firsteasylemma}
  and Lemma \ref{mainhardlemma}, we have that 
\begin{align*}
\sum_{0\le t_1<\cdots<t_d<n}  t_1^{r_1}&\cdots t_d^{r_d} \cdot \proj\big(\bX,\{t_1,\ldots,t_d\},c\big)
\\[-0.8ex] &\hspace{-0.2cm}=\sum_{0\le j_1<\cdots<j_d<k} \hspace{0.06cm}
w_{j_1,\ldots,j_d}\hspace{0.06cm} \sum_{0\le t_1<\cdots<t_d<n} g_{j_1,\ldots,j_d}(t_1,\ldots,t_d)\cdot \proj\big(
\bX,\{t_1,\ldots,t_d\},c\big) 
\end{align*}
can be written as a linear combination of entries of $\deck_k(\bX)$ in which each coefficient
  has magnitude at most $\smash{{k^{O(k)} \cdot {n\choose k}=n^{O(k)}}}$.
As a result, we have
$$
\left| \sum_{0\le t_1<\cdots<t_d<n} t_1^{r_1}\cdots t_d^{r_d} \cdot \Delta_{\bX,\bY,c}\big(\{t_1,\ldots,t_d\}\big) \right|
\le n^{O(k)}\cdot \|\deck_k(\bX)-\deck_k(\bY)\|_\infty.
$$
This finishes the proof of the lemma.
\end{proof}

Finally we prove Lemma \ref{mainhardlemma}.
We follow a three-step approach. We say that a \emph{quasimonomial} is a polynomial of the form $$t_1^{\alpha_1}\cdot  (t_2 - t_1-1)^{\alpha_2} \cdot (t_3 - t_2 - 1)^{\alpha_3}  \cdots (t_d - t_{d-1} - 1)^{\alpha_d} $$ for some nonnegative integers $\alpha_1,\ldots,\alpha_d$; the degree of this quasimonomial is $\alpha_1 + \cdots + \alpha_d$.  And we say that a \emph{PBC \emph{(}Product of Binomial Coefficients\emph{)}} is a polynomial of the form $${t_1 \choose \beta_1} {t_2 - t_1 -1 \choose \beta_2} \cdots {t_d - t_{d-1} - 1 \choose \beta_d}$$ for some nonnegative integers $\beta_1,\ldots,\beta_d$; the degree of this PBC is $\beta_1 + \cdots + \beta_d$. We observe that, compared to PBCs, the polynomials
  $g_{j_1,\ldots,j_d}$ have an extra binomial coefficient that involves $t_d$ at the end. The three steps of our approach are as follows:

\begin{flushleft}
\begin{itemize}
\item {\bf First step:}  Express each $d$-variable monomial $t_1^{r_1} \cdots t_d^{r_d}$ with
  $r_1+\cdots+r_d\le k-d$
  as a low-weight linear combination of quasimonomials of degree at most $k-d$.

\item {\bf Second step:}  Express each quasimonomial of degree at most $k-d$
 as a low-weight linear combination of PBCs of degree at most $k-d$.

\item {\bf Third step:}  Finally, express each PBC of degree at most $k-d$ as a 
low-weight linear combination of polynomials $g_{j_1,\ldots,j_d}$.
\end{itemize}
\end{flushleft}
We elaborate on each of these three steps below.
For each step, we bound the sum of magnitudes of coefficients in the linear combination.
\medskip

\noindent\textbf{First step.}
Consider the change of variables: $s_1 = t_1, s_2 = t_2 - t_1 - 1, \ldots, s_d = t_d - t_{d-1} - 1$. Then $$t_1^{r_1} t_2^{r_2} \cdots t_d^{r_d} = s_1^{r_1} (s_2 + s_1 + 1)^{r_2}  \cdots (s_d + s_{d-1} + \cdots + s_1 + d-1)^{r_d}.$$
Each monomial of $s_1,\ldots,s_d$ on the RHS corresponds to a quasimonomial of 
  degree at most $r_1+\cdots$ $+r_d\le k-d$ so this gives us an expression of $t_1^{r_1}\cdots t_d^{r_d}$
  as a linear combination of quasimonomials of degree at most $k-d$.
Moreover, the sum of magnitudes of the coefficients is bounded by
\begin{equation} \label{bound1} 
	3^{r_2} \cdot 5^{r_3} \cdots (2d - 1)^{r_d} \leq (2d - 1)^{\sum_{i = 2}^{d} r_i} \leq k^{O(k)}.
\end{equation}

\noindent\textbf{Second step.}
We start with a one-dimensional version of the second step.

\begin{claim} \label{claim:pascal}  
For each $r \geq 0$ and $t\ge 0$, we have \begin{equation} \label{eq:pascal}
t^r = \sum_{\beta=0}^r 
\left(
\sum_{i=0}^\beta (-1)^{\beta-i} \cdot {\beta \choose i} \cdot i^r 
\right)
{t \choose \beta}.
\end{equation}
\end{claim}

\begin{proof}
	We can write $\smash{t^r = \sum_{\beta = 0}^r v_{\beta} {t \choose \beta}}$ with $v \in \R^{r + 1}$ by changing bases in the space of polynomials in $t$.
	Let $P$ be the $\smash{(r+1) \times (r+1)}$ Pascal matrix with $\smash{P_{i, j} = {i \choose j}}$,
	and define $u \in \R^{r + 1}$ by $u_i = i^r$. Then $\smash{u = Pv}$ so $\smash{v = P^{-1} u}$. By Theorem 2 of \cite{BP92}, we have $$v_{\beta} = \sum_{i = 0}^{\beta} (-1)^{\beta - i} \cdot {\beta \choose i}\cdot i^r$$ as desired.
\end{proof}

We use Claim~\ref{claim:pascal} $d$ times to re-express each of $t_1^{\alpha_1}$, $(t_2-t_1-1)^{\alpha_2},\dots,(t_d-t_{d-1}-1)^{\alpha_d}$ as a linear combination of binomial coefficients. As a result, when $0\le t_1<\cdots<t_d<n$ we have

\begin{align*}
& t_1^{\alpha_1} \cdot
(t_2-t_1-1)^{\alpha_2}
\cdots
(t_d - t_{d-1} - 1)^{\alpha_d}\\[0.6ex]
&= 
\left(
\sum_{\beta_1=0}^{\alpha_1} 
\left(
\sum_{i_1=0}^{\beta_1} (-1)^{\beta_1-i_1} {\beta_1 \choose i_1} i_1^{\alpha_1}
\right)
{t_1 \choose \beta_1}
\right)
\cdot
\left(
\sum_{\beta_2=0}^{\alpha_2} 
\left(
\sum_{i_2=0}^{\beta_2} (-1)^{\beta_2-i_2} {\beta_2 \choose i_2} i_2^{\alpha_2}
\right)
{t_2 - t_1 - 1 \choose \beta_2}
\right)\\[0.5ex]
&\mbox{~~~~~~~~~~~}  \cdots
\left(
\sum_{\beta_d=0}^{\alpha_d} 
\left(
\sum_{i_d=0}^{\beta_d} (-1)^{\beta_d-i_d}  {\beta_d \choose i_d} i_d^{\alpha_d}
\right)
{t_d - t_{d-1} - 1 \choose \beta_d}
\right)\\[0.8ex]
&=\sum_{\beta_1,\dots,\beta_d} c_{\beta_1,\dots,\beta_d}\cdot  {t_1 \choose \beta_1} {t_2 - t_1 - 1 \choose \beta_2} \cdots {t_d - t_{d-1} - 1 \choose \beta_d}
\end{align*}
for coefficients $c_{\beta_1,\dots,\beta_d}$ that we will proceed to bound.
Note that the final sum is over $0\le \beta_i\le 
\alpha_i$, so this is a linear combination
  of PBCs of degree at most $\alpha_1+\cdots+\alpha_d\le k-d$.

Now we bound the sum of magnitudes of coefficients.
For $0 \leq \beta \leq \alpha$ we have $$ \left| \hspace{0.05cm}\sum_{i=0}^\beta (-1)^{\beta-i} \cdot {\beta \choose i}\cdot i^\alpha \hspace{0.05cm}\right| \leq \sum_{i=0}^\beta \beta^i i^\alpha \leq \sum_{i=0}^\beta (\beta i)^\alpha \le \beta^{O(\alpha) },$$ which implies (using $\alpha_1+\ldots+\alpha_d\le k-d\le k$)\begin{equation} \label{bound2} \sum_{\beta_1, \ldots, \beta_d} \left|  c_{\beta_1,\dots,\beta_d} \right| \leq \sum_{\beta_1, \ldots, \beta_d} \beta_1^{O(\alpha_1)} \cdots \beta_d^{O(\alpha_d)} \leq
\sum_{\beta_1,\ldots,\beta_d} k^{O(k)}= k^{O(k)}
\end{equation}


\noindent\textbf{Third step:} The next claim gives an  expression of a PBC as
 a linear combination
  of $g_{j_1,\ldots,j_d}$'s.

\begin{claim} [Third step:  $d$-variable combinatorial identity] \label{claim:IDd}
Given any $0\le t_1<\cdots<t_d<n$ and any nonnegative integers $\beta_1,\ldots,\beta_d$
  with $\beta_1+\cdots+\beta_d\le k-d$, we have 
\begin{align*}
\sum_{0 \leq j_1 <  \cdots < j_d<k} 
&g_{j_1,\dots,j_d}(t_1,\dots,t_d) \cdot 
{j_1 \choose \beta_1}
{j_2 - j_1 - 1 \choose \beta_2}
\cdots
{j_d - j_{d-1} - 1 \choose \beta_d}\\
&\hspace{-0.5cm}=
{t_1 \choose \beta_1} {t_2 - t_1 - 1 \choose \beta_2} \cdots 
{t_d - t_{d-1} - 1 \choose \beta_d} \cdot
{n - \beta_1 - \cdots - \beta_d - d \choose k - \beta_1 - \cdots - \beta_d - d}.
\end{align*}
\end{claim} 
%
\begin{proof}
Assume that $0\le t_1<\cdots<t_d<n$. We first consider the following combinatorial experiment with $n$ balls numbered $[0:n-1]$:  (1) \emph{Mark} $k$ of the $n$ balls, including balls  $t_1,\dots,t_d$;  (2) \emph{{Color red}} $\beta_1 + \cdots +\beta_d + d$ of the $k$ marked balls, including balls $t_1,\dots,t_d$, in such a way that for each $i \in [d]$, the $(\beta_1 + \cdots +\beta_i + i)$-th red one is $t_i$ (i.e., there are $\beta_1$ red balls before $t_1$, $\beta_2$
  red balls between $t_1$ and $t_2$, ..., and $\beta_d$ red balls between $t_{d-1}$ and $t_d$).
Below we count the total number of outcomes of this experiment (including which balls 
  are marked and which balls are colored red) in two different ways to obtain the desired identity.

In the first way, we
  note that at the end of this experiment there are $\beta_1$ balls that are \emph{{marked and red}} within the $t_1$ balls $\{0,\dots,t_1-1\}$ (and also $t_1$ is \emph{{marked and red}}), and for each $i \in [2:d]$ there are $\beta_i$ balls that are \emph{{marked and red}} within the $t_i - t_{i-1}-1$ balls $\{t_{i-1}+1,\dots,t_i-1\}$ (and also $t_i$ is \emph{{marked and red}}); and there are $k-\beta_1- \cdots - \beta_d - d$ other balls that are marked within the $n-\beta_1- \cdots - \beta_d - d$ other balls.  So the total number of outcomes of the experiment is precisely

\[
{t_1 \choose \beta_1} \cdot {t_2 - t_1 - 1 \choose \beta_2} \cdots 
{t_d - t_{d-1} - 1 \choose \beta_d} \cdot
{n - \beta_1 - \cdots - \beta_d - d \choose k - \beta_1 - \cdots - \beta_d - d}.
\]

We can also count the number of outcomes of the experiment in a different way, by viewing the experiment as being carried out as follows:
(a) For some numbers $0 \leq j_1 < \cdots < j_d <k$, \emph{mark} $k$ balls such that for each $i \in [d]$, the $(j_i+1)$-th   marked ball is $t_i$; note that {as mentioned earlier},
$g_{j_1,\dots,j_d}(t_1,\dots,t_d)$ is precisely the number of ways to 
do this.
(b) {Color red} $\beta_1$ of the $j_1$ marked balls before ball $t_1$ (and also color red ball $t_1$; there are ${j_1 \choose \beta_1}$ ways to do this), and for each $i \in [2:d]$, {color red} $\beta_i$ of the $j_i - j_{i-1}-1$ marked balls that lie in $\{t_{i-1}+1,\dots,t_i-1\}$ (and also color red  $t_i$; there are $\smash{j_i - j_{i-1} - 1 \choose \beta_i}$ ways to do this).
From this perspective, the total number of outcomes is 
\[
\sum_{0 \leq j_1 <  \cdots < j_d <k} 
g_{j_1,\dots,j_d}(t_1,\dots,t_d) \cdot 
{j_1 \choose \beta_1}
{j_2 - j_1 - 1 \choose \beta_2}
\cdots
{j_d - j_{d-1} - 1 \choose \beta_d},
\]
and we have established the identity
as desired.
\end{proof}

Observe that when $\beta_1+\cdots+\beta_d\le k-d$, we have
\begin{align*}\left| {j_1 \choose \beta_1}
{j_2 - j_1 - 1 \choose \beta_2}
\cdots
{j_d - j_{d-1} - 1 \choose \beta_d} \right| &\leq k^{\beta_1+\cdots+\beta_d}\le k^{k-d},
\end{align*}
 which implies that the sum of magnitudes of coefficients in the linear combination is\begin{align} \label{bound3} 
	 \frac{\sum_{0 \leq j_1 <  \cdots < j_d<k} \left| {j_1 \choose \beta_1}
{j_2 - j_1 - 1 \choose \beta_2}
\cdots
	 {j_d - j_{d-1} - 1 \choose \beta_d} \right|}{ {n - \beta_1 - \cdots - \beta_d - d \choose k - \beta_1 - \cdots - \beta_d - d} } &\leq \sum_{0 \leq j_1 <  \cdots < j_d<k} k^{k-d} \leq k^k.  \end{align}

We can now combine the three steps to prove Lemma \ref{mainhardlemma}. 
 
\begin{proof}[Proof of Lemma \ref{mainhardlemma}]
We express $t_1^{r_1} \cdots t_d^{r_d}$ as a linear combination of polynomials $g_{j_1,\ldots,j_d}$ 
  via the three steps as described above, with coefficients $w_{j_1,\ldots,j_d}$. 
The sum of magnitudes of coefficients in the linear combination used in the   
  First, Second, and Third steps are bounded from above 
  using inequalities (\ref{bound1}), (\ref{bound2}), and (\ref{bound3}) respectively.
These bounds give us $$\sum_{0 \leq j_1 <  \cdots < j_d <k} \left| w_{j_1, \ldots, j_d} \right| \leq k^{O(k)} \cdot k^{O(k )} \cdot k^k \leq k^{O(k )}$$ as desired. This finishes the proof of the lemma.
\end{proof}

%% file: worst-case-lower-bound.tex
\section{Lower bounds for distributions supported on at most ${2\ell}$ strings}\ignore{mixtures over $\ell$ strings} \label{sec:worst-case-lower}

Our main result in this section is Theorem~\ref{thm:worst-case-lower}, given below, which establishes a lower bound on the sample complexity of population recovery under the deletion channel which is exponential in the population size for a wide range of population sizes:
\ignore{
 Informally, this theorem says that any deletion-channel population recovery algorithm which has non-trivial performance for populations of size $\ell$ must use at least {$\Omega(n^{\ell})$} many samples.  
}
\begin{theorem} \label{thm:worst-case-lower} Fix any constant deletion probability $\delta \in (0, 1)$.
	Suppose that $A$ is an algorithm which, when run on i.i.d.~samples drawn from a distribution $\Del_\delta(\bX)$ with $|\supp(\bX)| \leq 2 {\ell}$, outputs a hypothesis $\tilde{\bX}$ which satisfies $\dtv(\bX,\tilde{\bX}) \leq 0.49$ with probability at least $0.51.$  Then
$A$ must use \[
{\frac {\Omega \left(n/{\ell^2} \right)^{\frac {\ell+1} 2}} {\ell^{\frac 3 2}}}
\]
 many samples.
\end{theorem}
{{If the population size upper bound $2 \ell$} is a constant this gives a lower bound of $\Omega(n^{(\ell+1)/2})$ samples, and for any $\ell<n^{0.499}$ this gives a lower bound of $n^{\Omega(\ell)}.$
}

For the rest of this section fix $\delta \in (0, 1)$ and let $\rho$ denote $1 - \delta$.  
The high-level idea of the proof is as follows:  We show that there exist two distributions $\bX,\bY$ over $\zo^n$ which have disjoint supports, each of size at most ${2 \ell}$, but satisfy
\begin{equation} \label{eq:wcl-key}
\dtv(\Del_{\delta}(\bX),\Del_{\delta}(\bY)) = 
	O \left( {\frac {\ell^2} n} \right)^{\frac {\ell+1} 2} \cdot \ell^{\frac 3 2} \cdot (1 - \delta)
\end{equation}
which clearly implies Theorem~\ref{thm:worst-case-lower}.

For simplicity throughout this section we assume that $n$ is odd, and we write $m$ to denote $(n-1)/2.$  The following notation will be useful:  For $0 \leq i \leq 2 \ell$ we write $e_{m+i}$ to denote the string $0^{m+i}10^{m-i}.$  The two distributions $\bX$ and $\bY$ that we consider will be supported on disjoint subsets of $\{e_{m+i}\}_{i\in[0 : 2 \ell]}$ (and hence each distribution has support size at most $2 \ell + 1$, but in our proofs neither will have full support so their support size will be at most $2 \ell$).

\medskip

\medskip

\noindent {\bf Notation and setup.} For notational convenience, let $\B(r)$ denote the binomial distribution $\Bin(r, {\rho})$.

Let $S$ be a set of indices, $\pi_S$ be a distribution over $S$, and $\{ \bV_i \}_{i \in S}$ be a set of random variables indexed by $S$. We write $\Mix(\pi_S ; \{ \bV_i \}_{i \in S})$ to denote the mixture over $\{ \bV_i \}_{i \in S}$ with each $\bV_i$ weighted by $\pi_S (i)$.

For conciseness we write $\bZ_n$ to denote a random variable which is distributed according to the binomial distribution $\B(n).$  We recall the following convenient expression for the falling moments of the binomial distribution:  for any $t = 0,1,\ldots$, we have
\begin{equation} \label{eq:falling-moments}
\E[\bZ_n (\bZ_n - 1) \cdots (\bZ_n - t)] = P_t(n), 
	\quad \quad \text{where~} P_t(n) =  {n(n-1)\cdots(n-t)} {{\rho}^{t+1}}.
\end{equation}

For completeness we include the derivation below: \begin{align*}
	\E[\bZ_n (\bZ_n - 1) \cdots (\bZ_n - t)] &= \sum_{i = 0}^n i(i-1)\cdots(i-t) \cdot {n \choose i} {\rho}^i (1-{\rho})^{n-i} \\
		&= \sum_{i = t+1}^n {\frac {n!} {(n-i)! (i - t - 1)!}} \cdot {\rho}^i \cdot (1 - {\rho})^{n - i} \\
		&= n(n-1)\cdots (n- t) {\rho}^{t+1} \sum_{j = 0}^{n - t - 1} {n - t - 1 \choose j} {\rho}^j (1 - {\rho})^{n - t - 1 - j} \\
		&= P_{t}(n).
\end{align*}

\medskip 


\noindent {\bf The key lemmas.}  The first main lemma makes precise the moment-matching property of $\pi_S$ and $\pi_T$ that we require:

\begin{lemma} [Matching moments of mixtures of disjointly supported binomial distributions] \label{lem:can-match-moments}
{Let $\ell \leq O(\sqrt{n})$.}\footnote{{Note that if $\ell = \omega(\sqrt{n})$ then Theorem~\ref{thm:worst-case-lower} holds trivially, so this assumption is without loss of generality.}}
There are two disjoint subsets $S,T \subset [0 : 2 \ell]$ and two distributions $\pi_S,\pi_T$ supported on $\{e_{m+i}\}_{i \in S}$ and $\{e_{m+j}\}_{j \in T}$ respectively with the following property (which we call the ``matching moment property''):

	 Let $\tilde{\bD}_S$ be a random variable whose distribution is the mixture of $\{\bZ_{m+i}\}_{i \in S}$ in which distribution $\bZ_{m+i}$ has mixing weight $\pi_S(e_{m+i})$, and likewise $\tilde{\bD}_T$ be a random variable whose distribution is the mixture of $\{\bZ_{m+j}\}_{j \in T}$ in which distribution $\bZ_{m+j}$ has mixing weight $\pi_T(e_{m+j})$. Then the first ${\ell}$ moments of $\tilde{\bD}_S$ and $\tilde{\bD}_T$ match each other, i.e. for all $t \in [\ell]$, we have
\begin{equation} \label{eq:matching-moments}
\E[(\tilde{\bD}_S)^t] = \E[(\tilde{\bD}_T)^t].
\end{equation}
\end{lemma}

The second main lemma (statement given in Lemma~\ref{lem:matching-moments-gives-lb} below) gives the desired upper bound on total variation distance.
To prove Theorem~\ref{thm:worst-case-lower} it suffices to prove Lemmas~\ref{lem:can-match-moments} and~\ref{lem:matching-moments-gives-lb}.

\subsection{Proof of Lemma~\ref{lem:can-match-moments} }

\begin{proof}
\ignore{In this proof, $k$ is fixed.}
{The proof is by a linear algebraic argument.}
Let $r = m + \ell$. Consider the mixtures 
\[
	\tilde{\bD}_S = \Mix(\{ a_{|i|} \}_{i \in [-\ell : \ell]} ; \{ \bZ_{r + i} \}_{i \in [-\ell : \ell]})
\] and \[
	\tilde{\bD}_T = \Mix(\{ b_{|j|} \}_{j \in [-\ell : \ell]} ; \{ \bZ_{r + j} \}_{j \in [-\ell : \ell]})
\] where all $a_i, b_i \in [0, 1]$ and {$\sum_{i=-\ell}^{\ell} a_{{|i|}} = 
	\sum_{j = -\ell}^{\ell} b_{{|j|}} = 1.$} Let $c_i = a_i - b_i$ for $0 \leq i \leq \ell$.

	We will prove the existence of a non-trivial solution $a_i, b_i$ (i.e., such that $a_i \neq b_i$ for some $i$) such that the following system holds:
\begin{equation} \label{SYSTEMA}
\begin{aligned}
\E[\tilde{\bD}_S] &= \E[\tilde{\bD}_T]\\
\E[\tilde{\bD}_S(\tilde{\bD}_S-1)] &= \E[\tilde{\bD}_T(\tilde{\bD}_T-1)] \\
& \cdots\\
\E[\tilde{\bD}_S(\tilde{\bD}_S-1)\cdots(\tilde{\bD}_S- \ell +1)] &= \E[\tilde{\bD}_T(\tilde{\bD}_T-1)\cdots(\tilde{\bD}_T- \ell +1)].
\end{aligned}
\end{equation} Observe that this is the same as requiring that $\E[\tilde{\bD}_S^t] = \E[\tilde{\bD}_T^t]$ for $t \leq \ell$. In (\ref{SYSTEMA}), we will be viewing $\E[Q(\tilde{\bD}_S)]$ and $\E[Q(\tilde{\bD}_T)]$ (for polynomials $Q$) as polynomials in $n$. We want the equations in~(\ref{SYSTEMA}) to hold as polynomial equalities.

Note that for $t \geq 0$, by (\ref{eq:falling-moments})
	we can rewrite the condition $\E[\tilde{\bD}_S(\bD_S-1)\cdots(\tilde{\bD}_S-t)] = \E[\tilde{\bD}_T(\tilde{\bD}_T-1)\cdots(\tilde{\bD}_T-t)] $ as the condition \begin{equation} \label{eq:valatn} c_0 P_t (r) + \sum_{i = 1}^{\ell} c_i \left( P_t (r +i) + P_t (r -i) \right) = 0, \end{equation} {viewing both sides as formal polynomials in $r$. Since $P_t(x)$ has degree $t+1$ in $x$, for $0 \leq \ell \leq t+1$ the coefficient of $r^\ell$ in the polynomial on the LHS of (\ref{eq:valatn}) is zero,}\ignore{By setting the coefficient of ${r}^\ell$ equal to $0$,\snote{Rephrase to imply that we consider the logical consequence of this equation involving only coefficients of $r^\ell$. Reviewer 1 interpreted this as saying we force the coefficient to be $0$.} for $0 \leq \ell \leq t+1$ (because $P_t (x)$ has degree $t+1$ in $x$),} and consequently we get a system of $t + 2$ homogeneous linear equations in $c_0, c_1, \ldots, c_{\ell}$.
		
		Naively, it seems that (\ref{SYSTEMA}) gives us $2 + 3 + \cdots + \ell + 1 = {{\ell +2} \choose 2} - 1$ many homogeneous linear equations, which is far more than the $\ell +1$ variables $c_0,\dots,c_{\ell}$  that are in play. At this point it is unclear that (\ref{SYSTEMA}) necessarily has a nonzero solution in the $c_i$'s. We will show that (\ref{SYSTEMA}) is actually comprised of at most $\ell$ equations and hence it must have a nonzero solution.

Thus, to prove the existence of a non-trivial solution to (\ref{SYSTEMA}) phrased in terms of $\tilde{\bD}_S$ and $\tilde{\bD}_T$, it suffices to prove the existence of a non-trivial solution to (\ref{SYSTEMA}) phrased in terms of polynomial equalities.

	To do this, we observe that equation (\ref{eq:valatn}) is also true when we replace $r$ by $r+1$ and get the condition \begin{equation} \label{eq:valatnp1} c_0 P_t (r +1) + \sum_{i = 1}^{\ell} c_i \left( P_t (r +1+i) + P_t (r +1-i) \right) = 0 \end{equation} as a polynomial in $r$. (Note that the initial assumption $\ell \leq \Omega(n)$ still holds if we increase $n$.)
	
Observe that
$$P_t (r +1) = P_t (r) + \rho \cdot (t + 1) P_{t-1} (r),$$
	so if we subtract (\ref{eq:valatn}) from (\ref{eq:valatnp1}) and divide by ${\rho}(t+1)$, then we get the condition \begin{equation*} c_0 P_{t-1} (r) + \sum_{i = 1}^{\ell} c_i \left( P_{t-1} (r +i) + P_{t-1} (r -i) \right) = 0 \end{equation*} as a polynomial in $r$. Since this is true for all $t$, then we have derived the condition $\E[\tilde{\bD}_S(\tilde{\bD}_S-1)\cdots(\tilde{\bD}_S-t+1)] = 0$. It follows by induction that all of (\ref{SYSTEMA}) follows from the condition $\E[\tilde{\bD}_S(\tilde{\bD}_S-1)\cdots(\tilde{\bD}_S- \ell +1)] = 0$.

		Thus we have reduced (\ref{SYSTEMA}) to a system of $\ell +1$ homogeneous linear equations over $\ell +1$ variables, but the first equation (which comes from {observing that the coefficient of $r^{\ell}$ in the LHS of (\ref{eq:valatn}) is 0}) will be \begin{equation} \label{eq:top}
			2c_{\ell} + 2c_{\ell-1} + \cdots + 2c_1 + c_0 = 0
\end{equation}
	and a second equation (which comes from {observing that the coefficient of $r^{\ell -1}$ in the LHS of (\ref{eq:valatn}) is 0}) will be $$ -\ell(\ell -1)c_{\ell} - \ell(\ell -1)c_{\ell -1} - \cdots - \ell(\ell -1) c_1 - (\ell /2)(\ell -1)c_0 = 0$$ because the coefficient of $r^{\ell -1}$ in $P_{\ell} (r)$ is $-(\ell /2)(\ell -1)$. This is just equation (\ref{eq:top}) times $-(\ell /2)(\ell -1)$. So there are actually at most $\ell$ distinct equations in $\ell +1$ variables, and hence there is (at least) a line of non-trivial solutions in the $c_i$'s.

	Given a satisfying assignment to the $c_i$'s, for each $i$ with $c_i = 0$ we set $a_i = b_i = 0$. If $c_i > 0$, then we set $a_i = c_i$ and $b_i = 0$. If $c_i < 0$, then we set $a_i = 0$ and $b_i = -c_i$. Note that \[
		0 = 2c_{\ell} + 2c_{\ell -1} + \cdots + 2c_1 + c_0 = 2a_{\ell} + 2a_{\ell -1} + \cdots + 2a_1 + a_0 - (2b_{\ell} + 2b_{\ell -1} + \cdots + 2b_1 + b_0) 	
	\] and that by homogeneity, we can scale all the $c_i$'s by any multiplicative constant and still get a valid solution to (\ref{SYSTEMA}). We scale the $c_i$'s so that $2a_{\ell} + 2a_{\ell -1} + \cdots + 2a_1 + a_0 = 1$. The above equation implies that $2b_{\ell} + 2b_{\ell -1} + \cdots + 2b_1 + b_0 = 1$ as well.
	
	This results in the coefficients $a_i$ and $b_i$ satisfying (\ref{SYSTEMA}) and $\tilde{\bD}_S$ and $\tilde{\bD}_T$ being valid distributions that are disjointly supported. Since the $c_i$'s were non-trivial, then at least one coefficient $c_i$ is non-zero and by equation (\ref{eq:top}), there exist coefficients $c_j$ and $c_k$ of opposite sign. Thus, both $\tilde{\bD}_S$ and $\tilde{\bD}_T$ have support sizes at most $2 \ell$.  

	We take $\pi_S (e_{m+t}) = a_{| t - \ell |}$ and $\pi_T (e_{m+t}) = b_{| t - \ell |}$ to conclude the proof.
\end{proof}

We will use the following corollary of Lemma~\ref{lem:can-match-moments}:

\begin{corollary} \label{cor:matching-moments}
Let $S,T,\pi_S,\pi_T$ be as in Lemma~\ref{lem:can-match-moments}.  Then for any polynomial $p$ of degree at most $\ell$, we have 
	
\begin{equation} \label{eq:polyeq}
\sum_{i \in \N} \pi_S(e_{m+i}) p(m+i)
= 
\sum_{j \in \N} \pi_S(e_{m+j}) p(m+j).
\end{equation}
\end{corollary}
\begin{proof}
Equation~(\ref{eq:matching-moments}) can be rewritten as
\[
\sum_{i \in \N} \pi_S(e_{m+i}) \E[(\bZ_{m+i})^t]=
\sum_{j \in \N} \pi_S(e_{m+j}) \E[(\bZ_{m+j})^t],
\]
which holds for all $t \leq \ell$. This is equivalent to having equal falling moments, i.e. for all $t \in [\ell],$
\[
	\sum_{i \in \N} \pi_S(e_{m+i}) \E[P_{t - 1} (\bZ_{m+i})]=
	\sum_{j \in \N} \pi_S(e_{m+j}) \E[P_{t - 1} (\bZ_{m+j})].
\] Indeed, for a random variable $\bZ$, $\E[P_{t - 1} (\bZ)]$ can be written as a linear combination of $1, \E[\bZ], \E[\bZ^2],$ $\ldots, $ $\E[\bZ^t]$ and since $1, P_{0} (\bZ), P_{1} (\bZ), \ldots, P_{\ell - 1} (\bZ)$ form a set of $\ell$ polynomials in $\bZ$ with degrees $0, 1, 2, \ldots, \ell$, then they form a basis for polynomials in $\bZ$ with degree at most $\ell$. 

By (\ref{eq:falling-moments}), this is in turn equivalent to having, for all $t \in [\ell],$
\[
	\sum_{i \in \N} \pi_S(e_{m+i}) \cdot P_{t - 1}(m+i) =
	\sum_{j \in \N} \pi_S(e_{m+j}) \cdot P_{t - 1}(m+j),
\] 
which is in turn equivalent to (\ref{eq:polyeq}) by the reasoning in the above paragraph.
\end{proof}

\subsection{Total Variation Distance Upper Bound} \label{sec:tvdub}

We state Lemma~\ref{lem:matching-moments-gives-lb} below. Informally, it says that if $\pi_S,\pi_T$ have the matching moment property, then the variation distance between two corresponding mixtures of two-dimensional vector-valued random variables is small.  (In the following, the notation $(\B(a),\B(b))$ stands for a vector-valued random variable in which the two coordinates are independently drawn from $\B(a)$ and $\B(b)$ respectively.)

\begin{lemma} \label{lem:matching-moments-gives-lb}
Let \ignore{$\pi_S,\pi_T$} {$\bX,\bY$} be two distributions with disjoint supports $\{e_{m+i}\}_{i \in S}$ and $\{e_{m+j}\}_{j \in T}$ respectively, where $S \cup T \subset [0 : 2 \ell]$, with the matching moment property from Lemma~\ref{lem:can-match-moments} above.  Then
\begin{equation} \label{eq:non-weak}
\ignore{
\dtv(\Mix(\pi_S; ((\B(m+i),\B(m-i)))_{i \in S}), \Mix(\pi_T;((\B(m+j),\B(m-j)))_{j \in T}))}
\dtv(\Del_{\delta}(\bX),\Del_{\delta}(\bY)) \leq O \left( {\frac {\ell^2} n} \right)^{\frac {\ell+1} 2} \cdot \ell^{\frac 3 2} {\cdot (1-\delta)}.
\end{equation}

\end{lemma}

\noindent {\bf Setup and useful results.}
Our proof of Lemma~\ref{lem:matching-moments-gives-lb} is based on ``moment-matching'' results for Poisson binomial distributions which were proved by Roos \cite{Roo00} and subsequently used by Daskalakis and Papadimitriou \cite{DP15}. Our approach is similar to the approach used in \cite{DP15}.
To state these results, recall that a \emph{Poisson binomial distribution} (PBD) is a sum $\bU = \bA_1 + \cdots + \bA_n$ of independent Bernoulli random variables (so each $\bA_i$ is a random variable taking value 1 with some probability $p_i \in [0,1]$ and taking value 0 with probability $1-p_i$). 

In \cite{DP15}, it is shown that if two PBDs satisfy some mild technical condition and have matching first $\ell$ moments, then they have total variation distance at most $2^{-\Omega(\ell)}$. We show that two mixtures of pairs of {suitable} binomially distributed variables that have matching first $\ell$ moments will have total variation distance at most  $n^{-\Omega(\ell)}$.

We recall Theorem~1 of \cite{DP15}, which gives a ``Krawtchouk expansion'' for any Poisson binomial distribution.  This provides an expression for the exact probability that the Poisson binomial distribution puts on any outcome in its support.  (We state the theorem for PBDs which are a sum of $n'$ many random variables, as when we apply it later it will be for such PBDs where $n' = m+\ell = (n-1)/2 + \ell.$)

\begin{theorem} [Theorem~1 of \cite{Roo00}, see also Theorem~7 of \cite{DP15}] \label{thm:roos}
Let $\bU = \bA_1 + \cdots + \bA_{n'}$ be a Poisson binomial distribution in which each $\bA_i$ takes value 1 with probability $p_i \in [0,1]$. Then for all $r \in [n']$ and all $p \in [0,1]$, we have
\begin{equation}
\Pr[\bU=r] =  \sum_{t=0}^{n'} \alpha_t(p_1,\dots,p_{n'}; p) \cdot \Delta^t B_{n' - t,p}(r),
\label{eq:roos}
\end{equation}
where
\begin{itemize}

\item $\alpha_0(p_1,\dots,p_{n'}; p) = 1$ and for $t \in [0 : n']$,
\[
\alpha_t(p_1,\dots,p_{n'}; p) := 
\sum_{1 \leq u(1) < \cdots < u(t) \leq n'} \prod_{r=1}^t (p_{u(r)}-p),
\]

\item and for all $t \in [0 : n']$,
\[
\Delta^t B_{n' - t,p}(r) := {\frac {(n'-t)!}{n'!}} \cdot {\frac {d^t}{dp^t}} B_{n',p}(r),
\]
where in the last expression $B_{n',p}(r)$ denotes the value ${n' \choose r} p^r (1-p)^{n'-r}$, the probability that the binomial distribution $\Bin(n',p)$ puts on the outcome $r$, viewed as a function of $p$.

\end{itemize}
\end{theorem}
\ignore{
The next corollary follows immediately from Theorem~\ref{thm:roos}.

\begin{corollary} 
\label{cor:roos}
Let $\bU = \bA_1 + \cdots + \bA_{n'}$ be a Poisson binomial distribution in which each $\bA_i$ takes value 1 with probability $p_i \in [0,1]$, and let 
$\bV = \bB_1 + \cdots + \bB_{n'}$ be an independent Poisson binomial distribution in which each $\bB_i$ takes value 1 with probability $q_i \in [0,1]$. Then for all $r,s \in [0 : n]$ and all $p \in [0,1]$, we have
\begin{equation}
\Pr[\bU=r \ \& \ \bV=s] =  
\left(
\sum_{t=0}^{n'} \alpha_t(p_1,\dots,p_{n'}; p) \cdot \Delta^t B_{n' - t,p}(r)
\right) 
\cdot
\left(
\sum_{t'=0}^{n'} \alpha_{t'}(q_1,\dots,q_{n'}; p) \cdot \Delta^{t'} B_{n' - t',p}(s)
\right),
\label{eq:roos2}
\end{equation}
where $\alpha_t(\cdot,\cdot)$ and $\Delta^t B_{n',p}(\cdot)$ are as in the statement of Theorem~\ref{thm:roos}.\end{corollary}}

We highlight the fact that $\Delta^t B_{n',p}(r)$ has no dependence on the parameters $p_1,\dots,p_{n'}$; this will be important for us later.

The following result, deduced from \cite{Roo00}, is very useful in analyzing (\ref{eq:roos}). It bounds each of the $n'+1$ summands in (\ref{eq:roos}) which add up to $\Pr[\bU=r]$.

\begin{theorem} 
\label{thm:roos-ub}
	Let $(p_1,\dots,p_{n'}) \in [0,1]^{n'}$, $p \in [0,1]$, and $\alpha_t(\cdot,\cdot)$ be as in the statement of Theorem~\ref{thm:roos}. Define \begin{equation} \label{eq:theta}
\theta(p_1,\dots,p_{n'}; p) :=
{\frac {2 \sum_{i=1}^{n'} (p_i - p)^2 + (\sum_{i=1}^{n'} (p_i - p))^2}
{2n'p^2 (1-p)^2}}.
\end{equation}

For $t \in [n']$, \begin{equation} \label{eq:roos-ub}
		|\alpha_t(p_1,\dots,p_{n'};p)| \cdot \|\Delta^t B_{n' - t,p}(\cdot)\|_1 \leq \sqrt{e} \cdot \theta(p_1,\dots,p_{n'};p)^{\frac t 2} t^{\frac 1 4}
\end{equation} where $\|\Delta^t B_{n' - t,p}(\cdot)\|_1$ denotes the 1-norm of $\Delta^t B_{n' - t,p}(\cdot)$ viewed as an $(n'+1)$-dimensional vector, i.e.
$\|\Delta^t B_{n' - t,p}(\cdot)\|_1 := \sum_{r=0}^{n'} \left| \Delta^t B_{n' - t,p}(r) \right|$.
\end{theorem}

\begin{proof}
	Inequality~(30) in \cite{Roo00} gives $$|\alpha_t(p_1,\dots,p_{n'};p)| \leq p^{\frac t 2} (1 - p)^{\frac t 2} \theta(p_1,\dots,p_{n'};p)^{\frac t 2} \left( {\frac {n'} {n' - t}} \right)^{\frac {n' - t} {2}}$$ for $t \in [n']$.

	Inequality~(38) in \cite{Roo00} gives $$\|\Delta^t B_{n' - t,p}(\cdot)\|_1 \leq \sqrt{e} \cdot t^{\frac 1 4} \left({\frac {n' - t} {n'}} \right)^{\frac {n' - t} {2}} \left( {\frac {t} {n' p (1-p)}} \right)^{\frac t 2}$$ for $t \in [n']$.

	By multiplying the above two inequalities together we get the desired result because $t \leq n'$.
\end{proof}

For conciseness we now let $\bD_S$ denote
$\Mix(\pi_S; ((\Bin(m+i,\rho),\Bin(m-i,\rho)))_{i \in S})$ where in each component two-dimensional distribution the two distributions $\Bin(m+i,\rho)$ and $\Bin(m-i,\rho)$ are independent, and similarly we let
$\bD_T$ denote $\Mix(\pi_T;((\Bin(m+j,\rho),\Bin(m-j,\rho)))_{j \in T})$. 
{In the rest of the proof we will argue that
\begin{equation}
\label{eq:goal}
\dtv(\bD_S,\bD_T) \leq O \left( {\frac {\ell^2} n} \right)^{\frac {\ell +1} 2} \cdot \ell^{\frac 3 2}
\end{equation}
This establishes the claimed upper bound on $\dtv(\Del_\delta(\bX),\Del_\delta(\bY))$ given in (\ref{eq:non-weak}).  To see this, observe that for any outcome in $\supp(\bX)$ or $\supp(\bY)$, with probability $\delta$ the one 1-coordinate is deleted under $\Del_\delta$ (in which case the distributions  resulting from $\Del_\delta(\bX)$ and $\Del_\delta(\bY)$ are identical), and that with the remaining $1-\delta$ probability (when the one 1-coordinate is not deleted) there is an exact correspondence between $\Del_\delta(\bX)$ and $\bD_S$ and between $\Del_\delta(\bY)$ and $\bD_T$.
}

For an index $c \leq n'$, let $v^{(c)}$ denote the $n'$-dimensional real vector whose first $c$ values are {$\rho$} and whose remaining values are $0$. 

For $t, t' \in [0:n']$ we define
\begin{align*}
C_{t, t'}(p) &= \sum_{i \in \N} \pi_S(e_{m+i}) \cdot \alpha_t (v^{(m+i)};p) \cdot \alpha_{t'}(v^{(m-i)};p),\\
D_{t, t'}(p) &= \sum_{j \in \N} \pi_T(e_{m+j}) \cdot \alpha_t (v^{(m+j)};p) \cdot \alpha_{t'}(v^{(m-j)};p).
\end{align*}

\noindent The following lemma is crucial for us.  Recall that $n'=m + \ell.$

\begin{lemma} \label{lem:DP6}
Let $\pi_S,\pi_T$ be as in the statement of Lemma~\ref{lem:matching-moments-gives-lb}. Then for any $p \in [0,1]$, the values $C_{t, t'}(p)$ and $D_{t, t'}(p)$ are identical for $t, t' \geq 0$ and $t + t' \leq \ell$.
\end{lemma}
\begin{proof}
Let $p$ be any value in $[0,1]$.  If $t + t' = 0$, then $t = t' = 0$. Recalling that $\alpha_0(\cdot, \cdot) \equiv 1$ we have that
\[
	C_{0, 0}(p) = \sum_{i \in \N} \pi_S(e_{m+i}) = 1 = \sum_{j \in \N} \pi_T(e_{m+j}) = D_{0,0}(p)
\]
as desired.  

	For $t + t' \geq 1,$ we observe that $\alpha_t(v^{(m+i)};p) \cdot \alpha_{t'}(v^{(m-i)};p)$ is composed of summands of the form $({\rho} - p)^{c + c'} (-p)^{t + t' - c - c'}$ for $c \in [0, t], c' \in [0, t']$. 
	
	In particular, we have \[
	C_{t, t'}(p) = \sum_{i \in \N} \pi_S(e_{m+i}) \cdot \sum_{c = 0}^t \sum_{c' = 0}^{t'} {m + i \choose c} {n' - m - i \choose t - c} {m - i \choose c'} {n' - m + i \choose t' - c'} \left({\rho} - p \right)^{c + c'} \left( -p \right)^{t + t' - c - c'},
\] in which each $\pi_S(e_{m+i})$ is multiplied by a polynomial in $m$ of degree at most $t + t' \leq \ell$. 
	
	Similarly, we have
\[
	D_{t, t'}(p) = \sum_{j \in \N} \pi_T(e_{m+j}) \cdot \sum_{c = 0}^t \sum_{c' = 0}^{t'} {m + j \choose c} {n' - m - j \choose t - c} {m - j \choose c'} {n' - m + j \choose t' - c'} \left({\rho} - p \right)^{c + c'} \left( -p \right)^{t + t' - c - c'}
\] and by Corollary~\ref{cor:matching-moments}, we see that $C_{t, t'}(p)=D_{t, t'}(p)$.
\end{proof}

Now we proceed to prove Lemma~\ref{lem:matching-moments-gives-lb}. Our argument is similar to the proof of Theorem~3 in \cite{DP15}.  

Let $p \in [0,1]$ and $r, s \in [0:n']$. We have
\begin{align*}
	\Pr[\bD_S =(r,s)]-\Pr[\bD_T =(r,s)] &= \sum_{t,t'=0}^{n'} \Delta^t B_{n',p} (r) \cdot \Delta^{t'} B_{n', p} (s) \left( C_{t, t'} (p) - D_{t, t'} (p) \right) \\
	&= \sum_{t + t' > \ell}^{n'} \Delta^t B_{n',p} (r) \cdot \Delta^{t'} B_{n', p} (s) \left( C_{t, t'} (p) - D_{t, t'} (p) \right)
\end{align*}
where the two lines are by Theorem~\ref{thm:roos} and Lemma~\ref{lem:DP6} respectively.
 As a result, for any $p \in [0,1]$ we have \begin{align*}
\dtv(\bD_S,\bD_T)  &= {\frac 1 2} \sum_{r, s=0}^{n'} |\Pr[\bD_S =(r,s)]-\Pr[\bD_T =(r,s)]| \nonumber\\
	&\leq {\frac 1 2} \sum_{t + t' > \ell}^{n'} |C_{t, t'}(p) - D_{t, t'}(p)| \cdot \|\Delta^t B_{n',p}(\cdot)\|_1 \cdot \|\Delta^{t'} B_{n',p}(\cdot)\|_1 \nonumber \\
	&\leq {\frac 1 2} \sum_{t + t' > \ell}^{n'} (|C_{t, t'}(p)| + |D_{t, t'}(p)|) \cdot \|\Delta^t B_{n',p}(\cdot)\|_1 \cdot \|\Delta^{t'} B_{n',p}(\cdot)\|_1 \nonumber
 \end{align*} and expanding out definitions gives \begin{align*}
	\dtv(\bD_S, \bD_T) &\leq {\frac 1 2} \sum_{t + t' > \ell}^{n'} \left( \left| \sum_{i \in S} \pi_S(e_{m+i}) \cdot \alpha_t(v^{(m+i)};p) \cdot \alpha_{t'}(v^{(m-i)};p) \right| \right. \nonumber\\
	&\indent \indent \indent + \left. \left|\sum_{j \in T} \pi_T(e_{m+j}) \cdot \alpha_t(v^{(m+j)};p) \cdot \alpha_{t'}(v^{(m-j)};p) \right| \right) \cdot \|\Delta^t B_{n',p}(\cdot)\|_1 \cdot \|\Delta^{t'} B_{n',p}(\cdot)\|_1 \nonumber\\ 
	&\leq {\frac 1 2} \sum_{i \in S} \pi_S(e_{m+i}) \sum_{t + t' > \ell}^{n'} \left|\alpha_t(v^{(m+i)};p) \right| \left|\alpha_{t'}(v^{(m-i)};p) \right| \cdot \|\Delta^t B_{n',p}(\cdot)\|_1 \cdot \|\Delta^{t'} B_{n',p}(\cdot)\|_1 \nonumber\\
	 &\indent \indent \indent + {\frac 1 2} \sum_{j \in T} \pi_T(e_{m+j}) \sum_{t + t' > \ell}^{n'} \left|\alpha_t(v^{(m+j)};p) \right| \left|\alpha_{t'}(v^{(m-j)};p) \right| \cdot \|\Delta^t B_{n',p}(\cdot)\|_1 \|\Delta^{t'} B_{n',p}(\cdot)\|_1.
\end{align*}

By Theorem~\ref{thm:roos-ub}, \begin{align}
	\dtv(\bD_S,\bD_T) &\leq {\frac e 2} \sum_{i \in S} \pi_S(e_{m+i}) \sum_{t + t' > \ell}^{n'} \theta(v^{(m+i)};\delta)^{\frac t 2} \cdot \theta(v^{(m-i)}; \delta)^{\frac {t'} 2} \cdot t^{\frac 1 4} {t'}^{\frac 1 4} \nonumber\\
	&\indent \indent \indent + {\frac e 2} \sum_{j \in T} \pi_T(e_{m+j}) \sum_{t + t' > \ell}^{n'} \theta(v^{(m+j)}; \delta)^{\frac t 2} \cdot \theta(v^{(m-j)}; \delta)^{\frac {t'} 2} \cdot t^{\frac 1 4} {t'}^{\frac 1 4} \nonumber\\ 
	&\leq {\frac e {2\sqrt{2}}} \sum_{i \in S} \pi_S(e_{m+i}) \sum_{t + t' > \ell}^{n'} {\theta(v^{(m+i)}; \delta)}^{\frac t 2} {\theta(v^{(m-i)}; \delta)}^{\frac {t'} 2} \sqrt{t + t'} \nonumber\\
	&\indent \indent \indent + {\frac e {2\sqrt{2}}} \sum_{j \in T} \pi_T(e_{m+j}) \sum_{t + t' > \ell}^{n'} {\theta(v^{(m+j)}; \delta)}^{\frac t 2} {\theta(v^{(m-j)}; \delta)}^{\frac {t'} 2} \sqrt{t + t'} \nonumber
\end{align} where the second inequality can be deduced from the AM-GM inequality.

Fix any $i \in S, j \in T$.  Let $p={\rho}$ in Theorem~\ref{thm:roos-ub}. Since $\rho$ is a constant in $(0,1)$ we get that 
\[
	\theta(v^{(m+i)};{\rho})= {\frac {2(\ell -i) {\rho}^2 + (\ell -i)^2 \cdot {\rho}^4} {2(m+ \ell) {\rho}^2 (1 - {\rho})^2}}
\leq O \left( {\frac {\ell^2} n} \right).
\] and similarly \[
	\theta(v^{(m-i)};{\rho}), \theta(v^{(m \pm j)};{\rho}) \leq O \left( \ell^2 / n \right)
	\] because $\ell \leq O(\sqrt{n})$ ({and we may assume that $\ell \leq O(\sqrt{n})$ since otherwise the total variation distance bound claimed in the lemma is trivial}).

By choosing sufficiently large $n$ and appropriate constants, we can upper bound the RHS by some $\theta < 1/2$. This gives

\begin{align*}
	\dtv(\bD_S,\bD_T) &\leq O \left( \sum_{t + t' > \ell}^{n'} {\theta}^{\frac {t + t'} 2} \sqrt{t + t'} \right) 
		\leq O \left( \sum_{i > \ell}^{n'} {\theta}^{\frac i 2} i^{\frac 3 2} \right) 
		\leq O(\ell +1)^{\frac {-1} 2} \sum_{i > \ell} {\theta}^{\frac i 2} i^2 
\end{align*} where the second inequality comes from the fact that there are $i + 1$ pairs of non-negative integers $t, t'$ that sum to $i$, and the third inequality comes from the fact that $i^{\frac 3 2} \leq (\ell +1)^{\frac {-1} 2} i^2$ when $i > \ell$.

Observe that \begin{align*}
	\sum_{i > \ell} x^i i^2 &= x \cdot {\frac d {dx}} (x \cdot {\frac d {dx}} {\sum_{i > \ell} x^i}) 
		= x \cdot {\frac d {dx}} (x \cdot {\frac d {dx}} {\frac {x^{\ell +1}} {1 - x}})
\end{align*} for $0 < x < 1$, so $$\sum_{i > \ell} x^i i^2 = {\frac {x^{\ell + 1}} {(1 - x)^3}} \cdot (\ell^2 (1 - x)^2 + 2 \ell (1 - x) + 1 + x) \leq O(\ell + 1)^2 \cdot {\frac {x^{\ell + 1}} {(1 - x)^3}} \leq O(\ell + 1)^2 \cdot x^{\ell + 1}$$ for $0 < x < 1/2$. This means \begin{align*}
	\dtv(\bD_S,\bD_T) &\leq O(\ell +1)^{\frac 3 2} {\theta}^{\frac {\ell +1} 2} 
		\leq O \left( {\frac {\ell^2} n} \right)^{\frac {\ell +1} 2} \cdot \ell^{\frac 3 2}
\end{align*}
giving (\ref{eq:goal}) as desired and concluding the proof of Lemma~\ref{lem:matching-moments-gives-lb}.

%% file: polynomial-h.tex

\section{Proof of Lemma~\ref{uuuuu}} \label{ap:polynomial-h}

\subsection{Chebyshev polynomials} \label{sec:chebyshev}

Let $T_r(x)$ denote the $r${th} Chebyshev polynomial of the first kind. Then $T_r$ has degree $r$
  and satisfies the following  property:

\begin{property}\label{lalu1}
	$T_r(1)=1$ and $|T_r(x)|\le 1$ for all $|x|\le 1$.
	 If $x>1$ then $T_r(x)>1$.
\end{property}

We will need an upper bound for $T_r(x)$ over $x\in [1,2]$.
For this purpose we recall the following explicit form of $T_r(x)$ for $|x|\ge 1$:
\begin{equation}\label{huuu}
T_r(x)=\frac{(x-\sqrt{x^2-1})^r+(x+\sqrt{x^2-1})^r}{2}.
\end{equation}

\begin{property}\label{hehehe}
	For $a\in [0,1]$, we have $T_r(1+a)\le e^{3r\sqrt{a}}$.
\end{property}
\begin{proof}
	Using (\ref{huuu}) we have
	$$ T_r(1+a)\le \left(1+a +\sqrt{2a+a^2}\right)^r
	\le (1+3\sqrt{a})^r\le e^{3r\sqrt{a} },$$
	where we used $a^2\le a\le \sqrt{a}$ when $a\in [0,1]$ and $1+x\le e^x$. 
\end{proof}

  The next property follows from the recurrence relation $$T_{r+1}(x)=2x\cdot T_r(x)-T_{r-1}(x)$$ {with initial conditions $T_0(x)=0$ and $T_1(x)=1.$ 

\begin{property}\label{boundcoefficients}
	 {For all $r\ge 0$, we have $\|T_r\|_1 \leq 3^r.$}\ignore{The sum of absolute values of coefficients of $T_r(x)$ is at most $3^r$.}
\end{property}


Following \cite{BEK99}, we write $g_r$ to denote the following polynomial of degree $r$:
$$
g_r(x)=\frac{1}{r+0.5}\cdot \left(\frac{T_0(x)}{2}+T_1(x)+\cdots+T_r(x)\right). 
$$

We need the following properties of the polynomial $g_r$.
Items 1, 2 and 3 of Property \ref{BEKPO} follow directly from Properties \ref{lalu1}, \ref{hehehe} and \ref{boundcoefficients}, respectively.
For item 4 we have
$$
g_r(\cos y)=\frac{1}{r+0.5}\cdot \big(0.5+\cos y+\cos 2y+\cdots+\cos ry\big)
=\frac{1}{r+0.5}\cdot \frac{\sin (r+0.5)y}{\sqrt{2(1-\cos y)}},
$$ 
for all $0<y\le \pi$. This implies that for all $x\in [-1,1)$, we have
$$
|g_r(x)|\le \frac{1}{r+0.5}\cdot \frac{1}{\sqrt{2(1-x)}}\le \frac{1}{r\sqrt{2(1-x)}}.
$$
\begin{property}\label{BEKPO}The polynomial $g_r$ satisfies the following properties.
	\begin{enumerate}
	\item $g_r(1)=1$ and $|g_r(x)|\le 1$ for all $|x|\le 1$;
	
	\item
	$1 \le g_r(1+a)\le e^{3r\sqrt{a}}$ for all $a\in [0,1]$;\vspace{-0.01cm} 
	\item $\|g_r\|_1 \leq 3^r$; and\vspace{-0.03cm}
	\item $
	|g_r(x)| 
	\le \frac{1}{r\sqrt{2(1-x)}}$ for all $x \in [-1,1)$. 
	\end{enumerate}
\end{property}

\subsection{Proof of Lemma~\ref{uuuuu}}

Recall the statement of Lemma~\ref{uuuuu}:

\medskip

\noindent {\bf Lemma~\ref{uuuuu}, restated.}  
\emph{
There is a univariate polynomial $h$ with the following properties:
\begin{enumerate}
\item $h$ has degree $O(\sqrt{m}\log m)$.
\item $h(0)=1$ and for each $b\in [m]$, 
$$|h(b)|\le \frac{1}{2^{\sqrt{b}}}\quad \text{and}\quad |h(-b)|\le e^{6\sqrt{b}\log m}.$$
\item  {$h$ satisfies $\|h\|_1 \leq \exp(O(\sqrt{m}\log m))$}.
\end{enumerate}
}

\medskip

\begin{proof}
Recall the polynomial $g_r$ in Section \ref{sec:chebyshev}.	
We use it to define a degree-$r$ polynomial $\psi_r$:
	$$
	\psi_r(x)=g_r\left(1-\frac{x}{m}\right).
	$$  
	Properties of $g_r$ directly imply the following properties of $\psi_r$:
	
\begin{enumerate}
\item $\psi_r(0)=1$;
\item For each $b\in [m]$, we have 
		$$
		|\psi_r(b)|\le \min\left(1, \frac{1}{r}\sqrt{\frac{m}{2b}}\right) \hspace{0.1cm}\quad\text{and}\quad
1\le \psi_r(-b)\le e^{3r\sqrt{b/m}}\hspace{0.05cm};$$ 
\item Finally, $\psi_r$ satisfies 
			$$
			\|\psi_r\|_1\le 3^r\cdot \left(1+\frac{1}{m}\right)^r.
			$$\end{enumerate} 

\def\tm{\tilde{m}}

Let $\tilde{m}=4^\beta$ be the smallest power of $4$ with $\tilde{m}\ge m$.
We use $\psi_r$ to define our  $h$ as follows:
$$
h(x)=\prod_{i\in [\beta]} \left(\psi_{\sqrt{\tm/4^{i-2}}}(x)\right)^{\sqrt{4^{i }}}.
$$  
First we have $h(0)=1$ and the degree of $h$ is at most
$$
\sum_{i\in [\beta]}  \sqrt{\frac{\tm}{4^{i-2}}}\cdot \sqrt{4^{i}}=O(\sqrt{\tm}\log_4 \tm)
=O(\sqrt{m}\log m).
$$
Next, given $b\in [m]$, let $i\in [\beta]$ be an integer such that 
   $4^{i-1}\le b\le 4^i$.
Then (using $\tm\ge m$)
$$
\left|\psi_{ \sqrt{\tm/4^{i-2}}}(b)\right|\le  \sqrt{\frac{4^{i-2}}{\tm}}\cdot \sqrt{\frac{m}{2b}}
\le \sqrt{\frac{4^{i-2}}{m}}\cdot \sqrt{\frac{m}{2\cdot 4^{i-1}}}<\frac{1}{2}.
$$
Using $|\psi_r(b)|\le 1$ for all $r$, we have that 
$$
|h(b)|\le  \frac{1}{2^{\sqrt{4^{i }}}}\le \frac{1}{ 2^{\sqrt{b}}}.
$$
On the other hand, we have for each $b\in [m]$ (using $\tm\le 4m$
  and that $m$ is asymptotically large),
$$
 h(-b)\le \exp\left(
3\sqrt{b/m}\sum_{i\in [\beta]} \sqrt{\tm/4^{i-2}}\cdot \sqrt{4^{i }}
\right)=\exp\left({24\sqrt{b}\hspace{0.06cm} \log_4 \tm }\right)
\le \exp\left(24\sqrt{b}\log m\right).
$$
Finally, the sum of magnitudes of coefficients of $h$ is at most
$$
\prod_{i\in [\beta]} \left(3^{\sqrt{\tm/4^{i-2}}}\cdot 2\right)^{\sqrt{4^{i}}}
=\exp\big(O(\sqrt{m}\log m)\big).
$$
This finishes the proof of the lemma.
\end{proof}